\documentclass{article}
\usepackage{amsmath}
\usepackage{xspace}
\usepackage{graphicx}
\usepackage{cite}
\usepackage{url}

\newtheorem{thm}{Theorem}[section]
\newtheorem{cor}[thm]{Corollary}
\newtheorem{obs}[thm]{Observation}

\newtheorem{prop}[thm]{Proposition}

\newtheorem{rem}[thm]{Remark}

\newtheorem{definition}{Definition}
\def\squarebox#1{\hbox to #1{\hfill\vbox to #1{\vfill}}}
\newcommand{\qed}{\hspace*{\fill}
\vbox{\hrule\hbox{\vrule\squarebox{.667em}\vrule}\hrule}\smallskip}
\newenvironment{proof}{\noindent{\bf Proof:~~}}{\(\qed\)}

\newcommand{\ie}{\emph{i.e.}}

\newcommand{\eg}{\emph{e.g.}}

\newcommand{\etal}{\emph{et al}}

\newcommand{\twoSPA}{{\sc 2nd-Price Auc\-tion$_{k}$}\xspace}
\newcommand{\twoSPhA}{{\sc 2nd-Price-Auc\-tion$_{k}$}\xspace}
\newcommand{\twoMP}{{\sc The Millionaires Prob\-lem$_{k}$}\xspace}

\newcommand{\bp}{{\sc Bisection Pro\-to\-col}\xspace}

\newcommand{\ba}{{\sc Bisection Auc\-tion}\xspace}
\newcommand{\bha}{{\sc Bisection-Auc\-tion}\xspace}
\newcommand{\bba}{{\sc Bisection Auc\-tion$_{g(k)}$}\xspace}
\newcommand{\bhba}{{\sc Bisection-Auc\-tion$_{g(k)}$}\xspace}

\newcommand{\pgk}{{\sc Public Good$_{k}$}\xspace}
\newcommand{\tpgkc}{{\sc Truthful Public Good$_{k,c}$}\xspace}

\begin{document}

\title{Approximate Privacy:\\ Foundations and Quantification}

\author{Joan Feigenbaum\thanks{Supported in part by NSF grants
0331548 and 0534052 and IARPA grant FA8750-07-0031.}\\
Dept.\ of Computer Science\\
Yale University\\
\texttt{joan.feigenbaum@yale.edu}\and Aaron D.\ Jaggard\thanks{Supported in part by
NSF grants 0751674 and 0753492.}\\
DIMACS\\
Rutgers University\\
\texttt{adj@dimacs.rutgers.edu}\and Michael
Schapira\thanks{Supported by NSF grant 0331548.}\\
Depts.\ of Computer Science\\
Yale University and UC Berkeley\\
\texttt{michael.schapira@yale.edu}}

\date{}

\maketitle

\begin{abstract}
Increasing use of computers and networks in business, government,
recreation, and almost all aspects of daily life has led to a
proliferation of online sensitive data about individuals and
organizations. Consequently, concern about the \emph{privacy} of
these data has become a top priority, particularly those data that
are created and used in electronic commerce. There have been many
formulations of privacy and, unfortunately, many negative results
about the feasibility of maintaining privacy of sensitive data in
realistic networked environments. We formulate
communication-complexity-based definitions, both worst-case and
average-case, of a problem's {\it privacy-approximation ratio}. We
use our definitions to investigate the extent to which approximate
privacy is achievable in two standard problems: the \emph{$2^{nd}$-price
Vickrey auction}~\cite{Vic61} and the \emph{millionaires problem} of
Yao~\cite{Yao-millionaires}.

For both the $2^{nd}$-price Vickrey auction and the millionaires
problem, we show that not only is perfect privacy impossible or
infeasibly costly to achieve, but even \emph{close approximations}
of perfect privacy suffer from the same lower bounds. By contrast,
we show that, if the values of the parties are drawn uniformly at
random from $\{0,\ldots,2^k-1\}$, then, for both problems, simple
and natural communication protocols have privacy-approximation
ratios that are linear in $k$ (\emph{i.e.}, logarithmic in the size
of the space of possible inputs). We conjecture that this improved
privacy-approximation ratio is achievable for \emph{any} probability
distribution.
\end{abstract}

\newpage

\section{Introduction}\label{sec-intro}

Increasing use of computers and networks in business, government,
recreation, and almost all aspects of daily life has led to a
proliferation of online sensitive data about individuals and
organizations. Consequently, the study of privacy has become a top
priority in many disciplines.  Computer scientists have contributed
many formulations of the notion of {\it privacy-preserving computation}
that have opened new avenues of investigation and shed new light on
some well studied problems.

One good example of a new avenue of investigation opened by concern
about privacy can be found in auction design, which was our original
motivation for this work.  Traditional auction theory is a central
research area in Economics, and one of its main questions is how to
incent bidders to behave truthfully, {\it i.e.}, to reveal
private information that auctioneers need in order to compute optimal
outcomes.  More recently, attention has turned to the complementary
goal of enabling bidders {\it not} to reveal private information that
auctioneers do {\it not} need in order to compute optimal outcomes.
The importance of bidders' privacy, like that of algorithmic efficiency,
has become clear now that many auctions are conducted online, and
Computer Science has become at least as relevant as Economics.

Our approach to privacy is based on communication complexity.
Although originally motivated by agents' privacy in mechanism
design, our definitions and tools can be applied to distributed
function computation in general.  Because perfect privacy can be
impossible or infeasibly costly to achieve, we investigate
approximate privacy.  Specifically, we formulate both worst-case and
average-case versions of the {\it privacy-approximation ratio} of a
function $f$ in order to quantify the amount of privacy that can be
maintained by parties who supply sensitive inputs to a distributed
computation of $f$. We also study the tradeoff between privacy
preservation and communication complexity.

Our points of departure are the work of Chor and
Kushilevitz~\cite{CK91} on characterization of privately computable
functions and that of Kushilevitz~\cite{K92} on the communication
complexity of private computation.  Starting from the same place,
Bar-Yehuda {\it et al.}~\cite{BCKO} also provided a framework in
which to quantify the amount of privacy that can be maintained in
the computation of a function and the communication cost of
achieving it. Their definitions and results are significantly different from the
ones we present here (see discussion in Appendix~\ref{apx:bar-yehuda}); as explained in
Section~\ref{section-discussion} below, a precise characterization
of the relationship between their formulation and ours is an
interesting direction for future work.

\subsection{Our Approach}

Consider an auction of a Bluetooth headset with $2$ bidders, $1$ and
$2$, in which the auctioneer accepts bids ranging from \$0 to \$7
in \$1 increments. Each bidder $i$ has a private value
$x_i\in\{0,\ldots,7\}$ that is the maximum he is willing to pay for
the headset. The item is sold in a $2^{nd}$-price Vickrey auction,
\emph{i.e.}, the \emph{higher} bidder gets the item (with ties
broken in favor of bidder $1$), and the price he pays is the
\emph{lower} bid. The demand for privacy arises naturally
in such scenarios~\cite{NPS99}: In a straightforward protocol, the
auctioneer receives sealed bids from both bidders and computes the
outcome based on this information. Say, \emph{e.g.}, that bidder $1$
bids \$3, and bidder $2$ bids \$6. The auctioneer sells the headset
to bidder $2$ for \$3. It would not be at all surprising however if,
in subsequent auctions of headsets in which bidder $2$ participates,
the same auctioneer set a reservation price of \$5. This could be
avoided if the auction protocol allowed the auctioneer to learn the
fact that bidder $2$ was the highest bidder (something he needs to
know in order to determine the outcome) but did not entail the full
revelation of $2$'s private value for the headset.

\begin{figure}[htp]
\begin{center}
\includegraphics[height=3.5in]{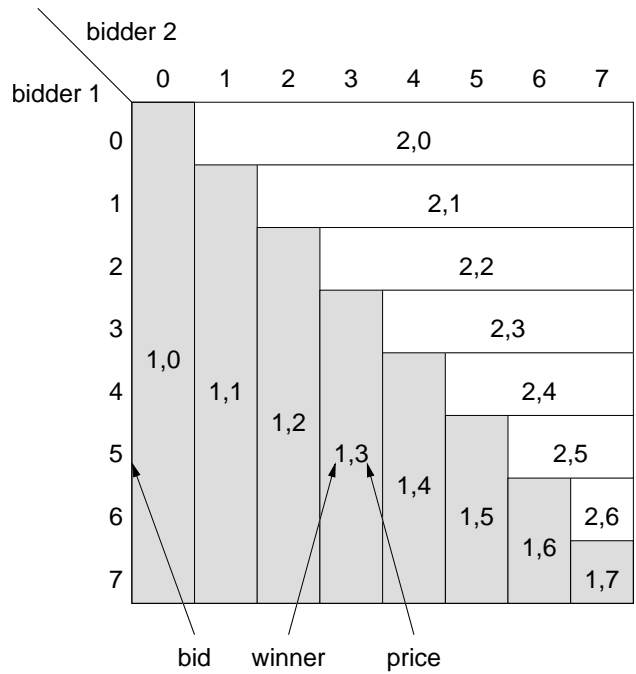}
\caption{\small The minimal knowledge requirements for
$2^{nd}$-price auctions}\label{fig:illustrate}
\end{center}
\end{figure}

Observe that, in some cases, revelation of the exact private
information of the highest bidder is necessary. For example,
if $x_1=6$, then bidder $2$ will win only if $x_2=7$. In other
cases, the revelation of a lot of information is necessary,
\emph{e.g.}, if bidder $1$'s bid is $5$, and bidder $2$ outbids him,
then $x_2$ must be either $6$ or $7$. An auction protocol is said
to achieve \emph{perfect objective privacy} if the auctioneer learns
nothing about the private information of the bidders that is not
needed in order to compute the result of the auction.
Figure~\ref{fig:illustrate} illustrates the information the auctioneer
\emph{must} learn in order to determine the outcome of the
$2^{nd}$-price auction described above. Observe that the auctioneer's failure to
distinguish between two potential pairs of inputs that belong to
different rectangles in Fig.~\ref{fig:illustrate} implies his
inability to determine the winner or the price the winner must pay.
Also observe, however, that the auctioneer need not be able to
distinguish between two pairs of inputs that belong to the
same rectangle.

Using the ``minimal knowledge requirements'' described in
Fig.~\ref{fig:illustrate}, we can now characterize a perfectly
(objective) privacy-preserving auction protocol as one that induces
this \emph{exact} partition of the space of possible inputs into
subspaces in which the inputs are indistinguishable to the
auctioneer. Unfortunately, perfect privacy is often hard or even
impossible to achieve. For $2^{nd}$-price auctions, Brandt and
Sandholm~\cite{BS} show that \emph{every} perfectly private auction
protocol has exponential communication complexity. This provides the
motivation for our definition of \emph{privacy-approximation ratio}:
We are interested in whether there is an auction protocol that
achieves ``good'' privacy guarantees without paying such a high
price in computational efficiency. We no longer insist that the
auction protocol induce a partition of inputs \emph{exactly} as in
Fig.~\ref{fig:illustrate} but rather that it ``approximate'' the
optimal partition well. We define two kinds of privacy-approximation
ratio (PAR): \emph{worst-case} PAR and \emph{average-case} PAR.

The worst-case PAR of a protocol $P$ for the $2^{nd}$-price auction
is defined as the maximum ratio between the size of a set $S$ of
indistinguishable inputs in Fig.~\ref{fig:illustrate} and the size
of a set of indistinguishable inputs induced by $P$ that is
contained in $S$. If a protocol is perfectly privacy preserving,
these sets are always the same size, and so the worst-case PAR is
$1$. If, however, a protocol fails to achieve perfect privacy, then at
least one ``ideal'' set of indistinguishable inputs \emph{strictly}
contains a set of indistinguishable inputs induced by the protocol.
In such cases, the worst-case PAR will be strictly higher than $1$.

Consider, \emph{e.g.}, the sealed-bid auction protocol in which both
bidders reveal their private information to the auctioneer, who then
computes the outcome. Obviously, this naive protocol enables the
auctioneer to distinguish between every two pairs of private
inputs, and so each set of indistinguishable inputs induced by the
protocol contains exactly one element. The worst-case PAR of
this protocol is therefore $\frac{8}{1}=8$. (If bidder $2$'s value
is $0$, then in Fig.~\ref{fig:illustrate} the auctioneer is unable to
determine which value in $\{0,\ldots,7\}$ is $x_1$. In the sealed
bid auction protocol, however, the auctioneer learns the exact value
of $x_1$.) The \emph{average-case} PAR is a natural Bayesian variant
of this definition: We now assume that the auctioneer has knowledge
of some market statistics, in the form of a probability
distribution over the possible private information of the bidders.
PAR in this case is defined as the average ratio and not as
the maximum ratio as before.

Thus, intuitively, PAR captures the effect of a protocol on the
privacy (in the sense of indistinguishability from other
inputs) afforded to protocol participants---it indicates the
factor by which, in the worst case or on average, using the
protocol to compute the function, instead of just being told
the output, reduces the number of inputs from which a given
input cannot be distinguished. To formalize and generalize the above intuitive definitions of PAR, we make
use of machinery from communication-complexity theory. Specifically,
we use the concepts of \emph{monochromaticity} and \emph{tilings} to
make formal the notions of sets of indistinguishable inputs and of
the approximability of privacy. We discuss other notions of
approximate privacy in Section~\ref{section-discussion}.

\subsection{Our Findings}

We present both upper and lower bounds on the privacy-approximation
ratio for both the millionaires problem and $2^{nd}$-price auctions
with $2$ bidders. Our analysis of these two environments takes place
within Yao's $2$-party communication model~\cite{Y79}, in which the
private information of each party is a $k$-bit string, representing
a value in $\{0,\ldots,2^k-1\}$. In the millionaires problem, the
two parties (the millionaires) wish to keep their private
information hidden \emph{from each other}. We refer to this
goal as the preservation of \emph{subjective} privacy. In
electronic-commerce environments, each party (bidder) often
communicates with the auctioneer via a secure channel, and so the
aim in the $2^{nd}$-price auction is to prevent \emph{a third
party} (the auctioneer), who is unfamiliar with any of the parties'
private inputs, from learning ``too much'' about the bidders. This
goal is referred to, in this paper, as the preservation of \emph{objective}
privacy.

Informally, for both the $2^{nd}$-price Vickrey auction and the
millionaires problem, we obtain the following results: We show that
not only is perfect privacy impossible or infeasibly costly to
achieve, but even close approximations of perfect privacy suffer
from the same lower bounds. By contrast, we show that, if the values
of the parties are drawn uniformly at random from
$\{0,\ldots,2^k-1\}$, then, for both problems, simple and natural
communication protocols have privacy-approximation ratios that are
linear in $k$ (\emph{i.e.}, logarithmic in the size of the space of
possible inputs). We conjecture
that this improved PAR is achievable for \emph{any} probability
distribution. The correctness of this conjecture would imply that,
no matter what beliefs the protocol designer may have about
the parties' private values, a protocol that achieves reasonable
privacy guarantees exists.

Importantly, our results for the $2^{nd}$-price Vickrey auction are
obtained by proving a more general result for a large family of
protocols for single-item auctions, termed ``\emph{bounded-bisection
auctions}'', that contains both the celebrated ascending-price
English auction and the class of bisection
auctions~\cite{GHM06,GHMV07}.

We show that our results for the millionaires problem also extend to
the classic economic problem of \emph{provisioning a public good},
by observing that, in terms of privacy-approximation ratios, the two
problems are, in fact, equivalent.

\subsection{Related Work: Defining Privacy-Preserving
Computation}

\subsubsection{Communication-Complexity-Based Privacy Formulations}

As explained above, the privacy work of Bar-Yehuda {\it et
al.}~\cite{BCKO} and the work presented in this paper have common
ancestors in \cite{CK91,K92}. Similarly, the work of Brandt and
Sandholm~\cite{BS} uses Kushilevitz's formulation to prove an
exponential lower bound on the communication complexity of
privacy-preserving $2^{nd}$-price Vickrey auctions. We elaborate on
the relation of our work to that of Bar-Yehuda {\it et
al.}~\cite{BCKO} in Appendix~\ref{apx:bar-yehuda}.

Similarly to~\cite{BCKO,CK91,K92}, our work focuses on the two-party
deterministic communication model. We view our results as first step in a more
general research agenda, outlined in Sec.~\ref{section-discussion}.

There are many formulations of privacy-preserving computation, both exact and approximate, that are not based on the
definitions and tools in~\cite{CK91,K92}. We now briefly review some of them and explain how they differ
from ours.

\subsubsection{Secure, Multiparty Function Evaluation}

The most extensively developed approach to privacy in distributed
computation is that of {\it secure, multiparty function evaluation}
(SMFE). Indeed, to achieve agent privacy in algorithmic mechanism
design, which was our original motivation, one could, in principle,
simply start with a strategyproof mechanism and then have the agents
themselves compute the outcome and payments using an SMFE protocol.
However, as observed by Brandt and Sandholm~\cite{BS}, these
protocols fall into two main categories, and both have inherent
disadvantages from the point of view of mechanism design:

\begin{itemize}
\item {\it Information-theoretically} private protocols, the study of which was
initiated by Ben-Or, Goldwasser, and Wigderson~\cite{BGW88} and
Chaum, Cr\'epeau, and Damgaard~\cite{CCD88}, rely on the assumption
that a constant fraction of the agents are ``honest'' (or
``obedient'' in the terminology of distributed algorithmic mechanism
design~\cite{FS02}), {\it i.e.}, that they follow the protocol
perfectly even if they know that doing so will lead to an outcome
that is not as desirable to them as one that would result from their
deviating from the protocol; clearly, this assumption is
antithetical to the main premise of mechanism design, which is that
all agents will behave strategically, deviating from protocols when
and only when doing so will improve the outcome from their points of
view;

\item Multiparty protocols that use {\it cryptography} to achieve
privacy, the study of which was initiated by
Yao~\cite{Yao-millionaires,Yao-general}, rely on (plausible but
currently unprovable) complexity-theoretic assumptions. Often, they
are also very communication-intensive (see, \emph{e.g.}, \cite{BS} for an explanation
of why some of the deficiencies of the Vickrey auction cannot be solved via cryptography).
Moreover, sometimes the deployment cryptographic machinery is infeasible (over the years,
many cryptographic variants of the current interdomain routing protocol, BGP, were proposed, but
not deployed due to the infeasibility of deploying a global Internet-wide PKI infrastructure
and the real-time computational cost of verifying signatures). For some mechanisms of interest,
efficient cryptographic protocols have been obtained (see, \eg, \cite{DHR00,NPS99}).
\end{itemize}

In certain scenarios, the demand for \emph{perfect} privacy preservation
cannot be relaxed. In such cases, if the function cannot be computed in a
privacy-preserving manner without the use of cryptography, there is no choice
but to resort to a cryptographic protocol. There is an extensive body of work
on cryptography-based identity protocols, and we are not offering our notion of PAR
as an extension of that work.  (In fact, the framework described here might be applied to SMFE protocols by replacing indistinguishability by computational indistinguishability.  However, this does not appear to yield any new insights.)

However, in other cases, we argue that privacy preservation
should be regarded as one of \emph{several design goals}, alongside
low computational/communication complexity, protocol
simplicity, incentive-compatibility, and more. Therefore, it is
necessary to be able to \emph{quantify} privacy preservation in order to
understand the tradeoffs among the different design goals, and
obtain ``reasonable'' (but not necessarily perfect) privacy guarantees.
Our PAR approach continues the long line of research about information-theoretic
notions of privacy, initiated by Ben-Or \etal.\ and by Chaum \etal.\
Regardless of the above argument, we believe that information-theoretic formulations of
privacy and approximate privacy are also natural to consider in
their own right.

\subsubsection{Private Approximations and Approximate Privacy}

In this paper, we consider protocols that compute exact results but
preserve privacy only approximately. One can also ask what it means
for a protocol to compute approximate results in a
privacy-preserving manner; indeed, this question has also been
studied~\cite{BCNW,FIMNSW,HKKN}, but it is unrelated to the
questions we ask here.  Similarly, definitions and techniques from
{\it differential privacy}~\cite{DiffPrivSurvey} (see also~\cite{GRS09}), in which the goal is to add noise
to the result of a database query in such a way as to preserve the
privacy of the individual database records (and hence protect the
data subjects) but still have the result convey nontrivial
information, are inapplicable to the problems that we study here.

\subsection{Paper Outline}

In the next section, we review and expand upon the connection
between perfect privacy and communication complexity. We present our
formulations of approximate privacy, both worst case and
average case, in Section~\ref{sec_approximate-privacy-def}; we present our
main results in Sections~\ref{section-bounds-I}
and~\ref{section-bounds-II}. Discussion and future directions can be
found in Section~\ref{section-discussion}.

\section{Perfect Privacy and Communication Complexity}\label{sec-model}

We now briefly review Yao's model of two-party communication and
notions of objective and subjective perfect privacy; see Kushilevitz
and Nisan~\cite{KN97} for a comprehensive overview of communication
complexity theory.  Note that we only deal with \emph{deterministic}
communication protocols. Our definitions can be extended to
randomized protocols.

\subsection{Two-Party Communication Model}\label{subsec_communication}

There are two parties, $1$ and $2$, each holding a $k$-bit
\emph{input string}. The input of party $i$, $x_i\in\{0,1\}^k$, is
the \emph{private information} of $i$. The parties communicate with
each other in order to compute the value of a function
$f:\{0,1\}^k\times \{0,1\}^k\rightarrow \{0,1\}^t$. The two parties
alternately send messages to each other. In communication round $j$, one of the parties sends a bit $q_j$ that is a
function of that party's input and the history $(q_1,\ldots,q_{j-1})$ of previously sent messages.  We say that a bit is \emph{meaningful} if it is not a constant function of this input and history and if, for every meaningful bit transmitted previously, there some combination of input and history for which the bit differs from the earlier meaningful bit.  Non-meaningful bits (\emph{e.g.}, those sent as part of protocol-message headers) are irrelevant to our work here and will be ignored.  A \emph{communication
protocol} dictates, for each party, when it is that party's turn to
transmit a message and what message he should transmit, based on the
history of messages and his value.

A communication protocol $P$ is said to compute $f$ if, for every
pair of inputs $(x_1,x_2)$, it holds that $P(x_1,x_2)=f(x_1,x_2)$. As
in~\cite{K92}, the last message sent in a protocol $P$ is assumed to
contain the value $f(x_1,x_2)$ and therefore may require up to $t$
bits. The \emph{communication complexity} of a protocol $P$ is the
maximum, over all input pairs, of the number of bits transmitted during the execution of $P$.

Any function $f:\{0,1\}^k\times \{0,1\}^k\rightarrow\{0,1\}^t$ can
be visualized as a $2^k\times 2^k$ matrix with entries in
$\{0,1\}^t$, in which the rows represent the possible inputs of
party $1$, the columns represent the possible inputs of party $2$,
and each entry contains the value of $f$ associated with its row and
column inputs. This matrix is denoted by $A(f)$.

\begin{definition}[Regions, partitions]
A \emph{region} in a matrix $A$ is any subset of entries in $A$ (not
necessarily a submatrix of $A$). A \emph{partition} of $A$ is a collection
of disjoint regions in $A$ whose union equals $A$.
\end{definition}

\begin{definition}[Monochromaticity]
A region $R$ in a matrix $A$ is called \emph{mono\-chro\-matic} if all
entries in $R$ contain the same value. A \emph{monochromatic
partition} of $A$ is a partition all of whose regions are
monochromatic.
\end{definition}

Of special interest in communication complexity are specific kinds
of regions and partitions called rectangles, and tilings,
respectively:

\begin{definition}[Rectangles, Tilings]\label{def:rect}
A \emph{rectangle} in a matrix $A$ is a submatrix of $A$. A
\emph{tiling} of a matrix $A$ is a partition of $A$ into rectangles.
\end{definition}

\begin{definition}[Refinements]
A tiling $T_1(f)$ of a matrix $A(f)$ is said to be a
\emph{refinement} of another tiling $T_2(f)$ of $A(f)$ if every
rectangle in $T_1(f)$ is contained in some rectangle in $T_2(f)$.
\end{definition}

Monochromatic rectangles and tilings are an important concept in
com\-mu\-ni\-ca\-tion-com\-plex\-i\-ty theory, because they are linked to the
execution of communication protocols. Every communication protocol
$P$ for a function $f$ can be thought of as follows:
\begin{enumerate}
\item Let $R$ and $C$ be the sets of row and column indices of $A(f)$, respectively.  For $R'\subseteq R$ and $C'\subseteq C$, we will abuse notation and write $R'\times C'$ to denote the submatrix of $A(f)$ obtained by deleting the rows not in $R'$ and the columns not in $C'$.

\item While $R\times C$ is not monochromatic:

\begin{itemize}

\item One party $i\in \{0,1\}$ sends a single bit $q$ (whose value
is based on $x_i$ and the history of communication).

\item If $i=1$, $q$ indicates whether $1$'s value is in one of
two disjoint sets $R_1,R_2$ whose union equals $R$. If $x_1\in R_1$,
both parties set $R=R_1$. If $x_1\in R_2$, both parties set $R=R_2$.

\item If $i=2$, $q$ indicates whether $2$'s value is in one of
two disjoint sets $C_1,C_2$ whose union equals $C$. If $x_2\in C_1$,
both parties set $C=C_1$. If $x_2\in C_2$, both parties set $C=C_2$.
\end{itemize}

\item One of the parties sends a last message (consisting of up to $t$ bits) containing the value in all
entries of the monochromatic rectangle $R\times C$.
\end{enumerate}

Observe that, for every pair of private inputs $(x_1,x_2)$, $P$
terminates at some monochromatic rectangle in $A(f)$ that contains
$(x_1,x_2)$. We refer to this rectangle as ``\emph{the monochromatic
rectangle induced by $P$ for $(x_1,x_2)$}''. We refer to the tiling
that consists of all rectangles induced by $P$ (for all pairs of
inputs) as  ``\emph{the monochromatic tiling induced by $P$}''.

\begin{figure}[htp]
\begin{center}
\includegraphics[width=1in]{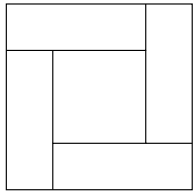}
\caption{\small A tiling that cannot be induced by any communication
protocol~\cite{K92}}\label{fig:nomech}
\end{center}
\end{figure}

\begin{rem}
There are monochromatic tilings that cannot be induced by communication
protocols. For example, observe that the tiling in
Fig.~\ref{fig:nomech} (which is essentially an example from~\cite{K92}) has this property.
\end{rem}

\subsection{Perfect Privacy}\label{subsec_perfect-privacy}

Informally, we say that a two-party protocol is \emph{perfectly
privacy-preserving} if the two parties (or a third party observing
the communication between them) cannot learn more from the execution
of the protocol than the value of the function the protocol
computes.  (These definition can be extended naturally to protocols involving more than two participants.)

Formally, let $P$ be a communication protocol for a function $f$.
The \emph{communication string} passed in $P$ is the concatenation
of all the messages $(q_1,q_2\ldots)$ sent in the course of the
execution of $P$. Let $s_{(x_1,x_2)}$ denote the communication
string passed in $P$ if the inputs of the parties are $(x_1,x_2)$.
We are now ready to define perfect privacy. The following two
definitions handle privacy from the point of view of a party $i$
that does not want the other party (who is, of course, familiar not
only with the communication string, but also with \emph{his own}
value) to learn more than necessary about $i$'s private information.
We say that a protocol is perfectly private with respect to party
$1$ if $1$ never learns more about party $2$'s private information
than necessary to compute the outcome.

\begin{definition} [Perfect privacy with respect to 1]~\cite{CK91,K92} \label{def_ppwrt1}
$P$ is \emph{perfectly private with respect to party $1$} if, for
every $x_2,x'_2$ such that $f(x_1,x_2)=f(x_1,x'_2)$, it holds that
$s_{(x_1,x_2)}=s_{(x_1,x'_2)}$.
\end{definition}

Informally, Def.~\ref{def_ppwrt1} says that party 1's knowledge of the communication
string passed in the protocol and his knowledge of $x_1$ do not
aid him in distinguishing between two possible inputs of $2$.
Similarly:
\begin{definition}[Perfect privacy with respect to 2]~\cite{CK91,K92}
$P$ is \emph{perfectly private with respect to party $2$} if, for
every $x_1,x'_1$ such that $f(x_1,x_2)=f(x'_1,x_2)$, it holds that
$s_{(x_1,x_2)}=s_{(x'_1,x_2)}$.
\end{definition}

\begin{obs}
For any function $f$, the protocol in which party $i$ reveals $x_i$
and the other party computes the outcome of the function is
perfectly private with respect to $i$.
\end{obs}

\begin{definition} [Perfect subjective privacy]
$P$ achieves \emph{perfect subjective privacy} if it is perfectly
private with respect to both parties.
\end{definition}

The following definition considers a different form of privacy---privacy from \emph{a third party} that observes the communication
string but has no \emph{a priori} knowledge about the private
information of the two communicating parties. We refer to this
notion as ``\emph{objective privacy}''.

\begin{definition} [Perfect objective privacy]
$P$ achieves \emph{perfect objective privacy} if, for every two
pairs of inputs $(x_1,x_2)$ and $(x'_1,x'_2)$ such that
$f(x_1,x_2)=f(x'_1,x'_2)$, it holds that
$s_{(x_1,x_2)}=s_{(x'_1,x'_2)}$.
\end{definition}

Kushilevitz~\cite{K92} was the first to point out the interesting
connections between perfect privacy and communication-complexity theory. Intuitively, we can think of any monochromatic
rectangle $R$ in the tiling induced by a protocol $P$ as a set of
inputs that are \emph{indistinguishable} to a third party. This is
because, by definition of $R$, for any two pairs of inputs in $R$, the
communication string passed in $P$ must be the same. Hence we can
think of the privacy of the protocol in terms of the tiling induced
by that protocol.

Ideally, every two pairs of inputs that are assigned the same
outcome by a function $f$ will belong to the same monochromatic
rectangle in the tiling induced by a protocol for $f$. This
observation enables a simple characterization of perfect
privacy-preserving mechanisms.

\begin{definition} [Ideal monochromatic partitions]
A monochromatic region in a matrix $A$ is said to be a {\em maximal
monochromatic region} if no monochromatic region in $A$ properly
contains it. \emph{The ideal monochromatic partition} of $A$ is made
up of the maximal monochromatic regions.
\end{definition}

\begin{obs}
For every possible value in a matrix $A$, the maximal mono\-chro\-matic
region that corresponds to this value is unique. This implies the
uniqueness of the ideal monochromatic partition for $A$.
\end{obs}

\begin{obs}\label{obs_privacy=maximal}\rm{\bf{(A characterization of perfectly privacy-preserving protocols)}}
A communication protocol $P$ for $f$ is perfectly privacy-preserving
iff the monochromatic til\-ing induced by $P$ is the ideal
monochromatic partition of $A(f)$. This holds for all of the above
notions of privacy.
\end{obs}

\section{Privacy-Approximation Ratios}\label{sec_approximate-privacy-def}

Unfortunately, perfect privacy should not be taken for granted. As
shown by our results, in many environments,
perfect privacy can be either impossible or very costly (in terms of
communication complexity) to obtain.
To measure a protocol's effect on privacy, relative to the
ideal---but perhaps impossible to implement---computation of the
outcome of a problem, we introduce the notion of
\emph{privacy-approximation ratios} (PARs).

\subsection{Worst-Case PARs}

For any communication protocol $P$ for a function $f$, we denote by
$R^P(x_1,x_2)$ the monochromatic rectangle induced by $P$ for
$(x_1,x_2)$. We denote by $R^I(x_1,x_2)$ the monochromatic region
containing $A(f)_{(x_1,x_2)}$ in the ideal monochromatic partition
of $A(f)$. Intuitively, $R^P(x_1,x_2)$ is the set of inputs that are
indistinguishable from $(x_1,x_2)$ to $P$. $R^I(x_1,x_2)$ is the set
of inputs that \emph{would be} indistinguishable from $(x_1,x_2)$ if
perfect privacy were preserved. We wish to asses how far one is from
the other. The size of a region $R$, denoted by $|R|$, is the
cardinality of $R$, \ie, the number of inputs in $R$.

We can now define worst-case \emph{objective} PAR as follows:
\begin{definition}[Worst-case objective PAR of $P$]
The \emph{worst-case objective privacy-ap\-prox\-i\-ma\-tion ratio} of
communication protocol $P$ for function $f$ is
\[
    \alpha=\max_{(x_1,x_2)}\ \frac{|R^I(x_1,x_2)|}{|R^P(x_1,x_2)|}.
\]

We say that $P$ is \emph{$\alpha$-objective-privacy-preserving in
the worst case}.
\end{definition}

\begin{definition} [$i$-partitions]
The \emph{$1$-partition} of a region $R$ in a matrix $A$ is the set of
disjoint rectangles $R_{x_1}=\{x_1\}\times\{x_2\ s.t.\
(x_1,x_2)\in R\}$ (over all possible inputs $x_1$).
\emph{$2$-partitions} are defined analogously.
\end{definition}

Intuitively, given any region $R$ in the matrix $A(f)$, if party
$i$'s actual private information is $x_i$, then $i$ can use this
knowledge to eliminate all the parts of $R$ other than $R_{x_i}$.
Hence, the other party should be concerned not with $R$ but rather
with the $i$-partition of $R$.

\begin{definition} [$i$-induced tilings]
The \emph{$i$-induced tiling} of a protocol $P$ is the refinement of the
tiling induced by $P$ obtained by $i$-partitioning each rectangle in
it.
\end{definition}

\begin{definition} [$i$-ideal monochromatic partitions]
The \emph{$i$-ideal monochromatic partition} is the refinement of the ideal
monochromatic partition obtained by $i$-partitioning each region in
it.\end{definition}

For any communication protocol $P$ for a function $f$, we use
$R_i^P(x_1,x_2)$ to denote the monochromatic rectangle containing
$A(f)_{(x_1,x_2)}$ in the $i$-induced tiling for $P$. We denote by
$R_i^I(x_1,x_2)$ the monochromatic rectangle containing
$A(f)_{(x_1,x_2)}$ in the $i$-ideal monochromatic partition of
$A(f)$.

\begin{definition}[Worst-case PAR of $P$ with respect to $i$]
\ The \emph{worst-case privacy-ap\-prox\-i\-ma\-tion ratio with respect to
$i$} of communication protocol $P$ for function $f$ is
\[
    \alpha=\max_{(x_1,x_2)}\ \frac{|R_i^I(x_1,x_2)|}{|R^P_i(x_1,x_2)|}.
\]

We say that $P$ is \emph{$\alpha$-privacy-preserving with respect to
$i$ in the worst case}.
\end{definition}

\begin{definition}[Worst-case subjective PAR of $P$]
The \emph{worst-case subjective privacy-ap\-prox\-i\-ma\-tion ratio} of
communication protocol $P$ for function $f$ is the maximum of the
worst-case privacy-approximation ratio with respect to each party.
\end{definition}

\begin{definition}[Worst-case PAR]
The \emph{worst-case objective (subjective) PAR for a function $f$}
is the minimum, over all protocols $P$ for $f$, of the worst-case objective (subjective) PAR of $P$.
\end{definition}

\subsection{Average-Case PARs}

As we shall see below, it is also useful to define an average-case
version of PAR. As the name suggests, the average-case objective PAR
is the \emph{average} ratio between the size of the monochromatic
rectangle containing the private inputs and the corresponding region
in the ideal monochromatic partition.

\begin{definition}[Average-case objective PAR of $P$]
Let $D$ be a probability distribution over the space of inputs. The
\emph{average-case objective privacy-approximation ratio} of
communication protocol $P$ for function $f$ is
\[
    \alpha = E_{D}\ [\frac{|R^I(x_1,x_2)|}{|R^P(x_1,x_2)|}].
\]

We say that $P$ is \emph{$\alpha$-objective privacy-preserving in
the average case with distribution $D$} (or \emph{with respect to
$D$}).
\end{definition}

We define average-case PAR with respect to $i$ analogously, and average-case subjective PAR as the maximum over all players $i$ of the average-case PAR with respect to $i$.
We define the \emph{average-case objective (subjective) PAR for a function
$f$} as the minimum, over all protocols $P$ for $f$, of the average-case objective (subjective) PAR of $P$.

\section{The Millionaires Problem and Public Goods: Bounds on PARs}\label{section-bounds-I}

In this section, we prove upper and lower bounds on the
privacy-approximation ratios for two classic problems: Yao's
millionaires problem and the provision of a public good.

\subsection{Problem Specifications}

\paragraph*{The millionaires problem.} Two millionaires want to know which one is richer.
Each millionaire's wealth is private information known only to him,
and the millionaire wishes to keep it that way. The goal is to
discover the identity of the richer millionaire while preserving the
(subjective) privacy of both parties.

\begin{definition}
[\twoMP]
\noindent \\
\noindent \underline{Input:} $x_1,x_2\in \{0,\ldots,2^k-1\}$ (each represented by a $k$-bit string)\\
\noindent \underline{Output:} the identity of the party with the
higher value, \emph{i.e.}, $\arg\max_{i\in \{0,1\}} x_i$ (breaking
ties lexicographically).
\end{definition}

There cannot be a perfectly privacy-preserving communication protocol for \twoMP
~\cite{K92}. Hence, we are interested in the PARs for this
well studied problem.

\paragraph*{The public-good problem.} There are two agents, each
with a private value in $\{0,\ldots,2^k-1\}$ that represents his
benefit from the construction of a public project (public good),
\emph{e.g.}, a bridge.\footnote{This is a discretization of the
classic public good problem, in which the private values are taken
from an interval of reals, as in~\cite{BN02,BBNS08}.} The goal of
the social planner is to build the public project only if the sum of
the agents' values is at least its cost, where, as in~\cite{BBNS08},
the cost is set to be $2^k-1$.

\begin{definition}
[\pgk]

\noindent \\
\noindent \underline{Input:} $x_1,x_2\in \{0,\ldots,2^k-1\}$ (each represented by a $k$-bit string)\\
\noindent \underline{Output:} ``Build'' if $x_1+x_2\geq 2^k-1$, ``Do
Not Build'' otherwise.
\end{definition}

It is easy to show (via
Observation~\ref{obs_privacy=maximal}) that for \pgk, as for \twoMP,
no perfectly privacy-preserving communication protocol exists.
Therefore, we are interested in the PARs for this problem.

\subsection{The Millionaires Problem}

The following theorem shows that not only is perfect subjective
privacy unattainable for \twoMP, but a stronger result holds:
\begin{thm} [A worst-case lower bound on subjective PAR]
\ \ No communication protocol for \twoMP has a worst-case subjective
PAR less than $2^{\frac{k}{2}}$.
\end{thm}
\begin{proof}
Consider a communication protocol $P$ for \twoMP. Let $R$ represent
the space of possible inputs of millionaire $1$, and let $C$
represent the space of possible inputs of millionaire $2$. In the
beginning, $R=C=\{0,\ldots,2^k-1\}$. Consider the first (meaningful)
bit $q$ transmitted in course of $P$'s execution. Let us assume that
this bit is transmitted by millionaire $1$. This bit indicates
whether $1$'s value belongs to one of two disjoint subsets of $R$,
$R_1$ and $R_2$, whose union equals $R$. Because we are interested in the worst case, we can choose adversarially to
which of these subsets $1$'s input belongs. Without loss of
generality, let $0\in R_1$. We decide adversarially that $1$'s value
is in $R_1$ and set $R=R_1$. Similarly, if $q$ is transmitted by
millionaire $2$, then we set $C$ to be the subset of $C$ containing
$0$ in the partition of $2$'s inputs induced by $q$. We continue
this process recursively for each bit transmitted in $P$.

Observe that, as long as both $R$ and $C$ contain at least two values,
$P$ is incapable of computing \twoMP. This is because $0$ belongs to
both $R$ and $C$, and so $P$ cannot eliminate, for either of the
millionaires, the possibility that that millionaire has a value of $0$
and the other millionaire has a positive value. Hence, this process
will go on until $P$ determines that the value of one of the
millionaires is exactly $0$, \emph{i.e.}, until either $R=\{0\}$ or
$C=\{0\}$. Let us examine these two cases:
\begin{itemize}
\item {\bf Case I: $R=\{0\}$.} Consider the subcase in which
$x_2$ equals $0$.  Recall that $0\in C$, and so this is possible.
Observe that, in this case, $P$ determines the exact value of $x_1$,
despite the fact that, in the $2$-ideal-monochromatic partition, all $2^k$
possible values of $x_1$ are in the same monochromatic rectangle when
$x_2=0$ (because for all these values $1$ wins). Hence, we get a
lower bound of $2^k$ on the subjective privacy-approximation ratio.

\item {\bf Case II: $C=\{0\}$.} Let $m$ denote the highest input in $R$.
We consider two subcases.  If $m\leq 2^{\frac{k}{2}}$, then observe
that the worst-case subjective privacy-approximation ratio is at
least $2^{\frac{k}{2}}$.  In the $2$-ideal-monochromatic partition,
all $2^k$ possible values of $x_1$ are in the same monochromatic
rectangle if $x_2=0$, and the fact that $m\leq 2^{\frac{k}{2}}$
implies that $|R|\leq 2^{\frac{k}{2}}$.

If, on the other hand, $m>2^{\frac{k}{2}}$, then consider the case
in which $x_1=m$ and $x_2=0$. Observe that, in the
$1$-ideal-monochromatic partition, all values of millionaire $2$ in
$\{0,\ldots,m-1\}$ are in the same monochromatic rectangle if
$x_1=m$. However, $P$ will enable millionaire $1$ to determine that millionaire $2$'s value
is exactly $0$. This implies a lower bound of $m$ on the subjective
privacy-approximation. We now use the fact that $m>2^{\frac{k}{2}}$
to conclude the proof.
\end{itemize}
\end{proof}

By contrast, we show that fairly good privacy guarantees can be
obtained in the average case. We define the \bp for \twoMP as
follows: Ask each millionaire whether his value lies in
$[0,2^{k-1})$ or in $[2^{k-1},2^k)$; continue this binary search
until the millionaires' answers differ, at which point we know which
millionaire has the higher value. If the answers never differ the
tie is broken in favor of millionaire $1$.

We may exactly compute the average-case subjective PAR with respect to the uniform distribution for the \bp applied to \twoMP.  Figure~\ref{fig:mp-bp-spar} illustrates the approach.  The far left of the figure shows the ideal partition (for $k=3$) of the value space for \twoMP; these regions are indicated with heavy lines in all parts of the figure.  The center-left shows the $1$-partition of the regions in the ideal partition; the center-right shows the $1$-induced tiling that is induced by the \bp.  The far right illustrates how we may rearrange the tiles that partition the bottom-left region in the ideal partition (by reflecting them across the dashed line) to obtain a tiling of the value space that is the same as the tiling induced by applying the \ba to \twoSPA.
\begin{figure*}[htp]
\begin{center}
\includegraphics[width=4.75in]{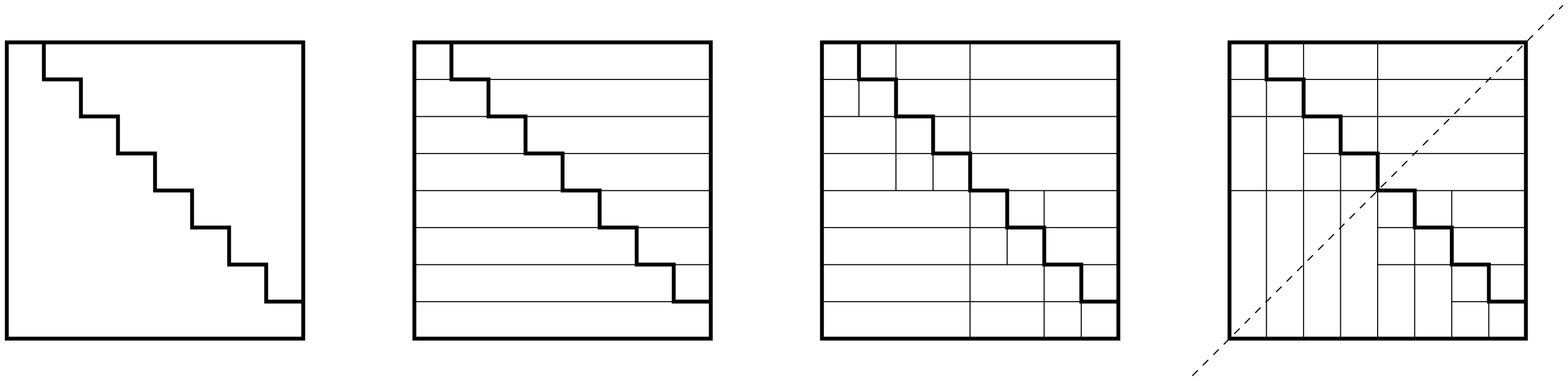}
\caption{Left to right: The ideal partition (for $k=3$) for \twoMP; the $1$-partition of the ideal regions; the $1$-induced tiling induced by the \bp; the rearrangement used in the proof of Thm.~\ref{thm:mp-bp-spar}}\label{fig:mp-bp-spar}
\end{center}
\end{figure*}

\begin{thm}\label{thm:mp-bp-spar}\rm{\bf{The average-case subjective PAR of the bisection protocol}}
The av\-er\-age-case subjective PAR with respect
to the uniform distribution for the \bp applied to \twoMP is
$\frac{k}{2}+1$.
\end{thm}
\begin{proof}
Given a value of $i$, consider the $i$-induced-tiling
obtained by running the \bp for \twoMP (as in the center-right of Fig.~\ref{fig:mp-bp-spar} for $i=1$).
Rearrange the rectangles in which player $i$ wins by reflecting them
across the line running from the bottom-left corner to the top-right
corner (the dashed line in the far right of Fig.~\ref{fig:mp-bp-spar}).  This produces a tiling of the value space in which the
region in which player $1$ wins is tiled by tiles of width $1$, and
the region in which player $2$ wins is tiled by tiles of height $1$;
in computing the average-case-approximate-privacy with respect to
$i$, the tile-size ratios that we use are the heights (widths) of
the tiles to the height (width) of the tile containing all values in
that column (row) for which player $1$ ($2$) wins.  This tiling and
the tile-size ratios in question are exactly as in the computation
of the average-case \emph{objective} privacy for \twoSPA; the
argument used in Thm.~\ref{thm:2spa-par-bba} (for $g(k)=k$) below completes the proof.
\end{proof}

Consider the case in which a third party is observing the
interaction of the two millionaires. How much can this observer
learn about the private information of the two millionaires? We show
that, unlike the case of subjective privacy, good PARs are
unattainable even in the average case.

Because the values $(i,i)$ (in which case player $1$ wins) and the values $(i,i+1)$ (in which player $2$ wins) must all appear in different tiles in any tiling that refines the ideal partition of the value space for \twoMP, any such tiling must include at least $2^k$ tiles in which player $1$ wins
and $2^k-1$ tiles in which player $2$ wins.
The total contribution of a tile in which player $1$ wins is the
number of values in that tile times the ratio of the ideal region
containing the tile to the size of the tile, divided by the total
number ($2^{2k}$) of values in the space.  Each tile in which player
$1$ wins thus contributes $\frac{(1+2^k)2^k}{2^{2k+1}}$ to the
average-case PAR under the uniform
distribution; similarly, each tile in which player $2$ wins
contributes $\frac{2^k(2^k-1)}{2^{2k+1}}$ to this quantity. This leads
directly to the following result.
\begin{prop}[A lower bound on average-case objective PAR]
The average-case objective PAR for \twoMP with respect to the uniform distribution is at least $2^k - \frac{1}{2} + 2^{-(k+1)}$.
\end{prop}

There are numerous different tilings of the value space that achieve
this ratio and that can be realized by communication protocols.  For
the \bp, we obtain the same exponential (in $k$) growth rate but
with a larger constant factor.

\begin{prop}\rm{\bf{(The average-case objective PAR of the bisection protocol)}}
The \bp for \twoMP obtains an average-case objective
PAR of $3\cdot 2^{k-1} - \frac{1}{2}$ with
respect to the uniform distribution.
\end{prop}
\begin{proof}
The bisection mechanism induces a tiling that refines the ideal
partition and that has $2^{k+1}-1$ tiles in which the player $1$
wins and $2^k-1$ tiles in which the player $2$ wins.  The
contributions of each of these tiles is as noted above, from which
the result follows.
\end{proof}

Finally, Table~\ref{tab:twomp} summarizes our average-case PAR results (with respect to the uniform distribution) for \twoMP.
\begin{table}[htp]
\begin{center}
\begin{tabular}{|c|c|c|}
  \hline
   & Average-Case Obj.\ PAR & Average-Case Subj.\ PAR \\ \hline
  Any Protocol & $\geq 2^k - \frac{1}{2} + 2^{-(k+1)}$ &  \\ \hline
  Bisection Protocol & $\frac{3}{2}2^k - \frac{1}{2}$ & $\frac{k}{2}+1$ \\
  \hline
\end{tabular}
\end{center}
\caption{Average-case PARs for \twoMP}\label{tab:twomp}
\end{table}

\subsection{The Public-Good Problem}\label{ssec:pg}

The government is considering the construction of a bridge (a public good) at cost $c$.  Each taxpayer has a $k$-bit private value that is the utility he would gain from the bridge if it were built.  The government wants to build the bridge if and only if the sum of the taxpayers' private values is at least $c$.  In the case that $c=2^k-1$, we observe that $\hat{x}_2 = c-x_2$ is again a $k$-bit value and that $x_1+x_2 \geq c$ if and only if $x_1 \geq \hat{x}_2$; from the perspective of PAR, this problem is equivalent to solving \twoMP\ on inputs $x_1$ and $\hat{x}_2$.  We may apply our results for \twoMP to see that the public-good problem with $c=2^k-1$ has exponential average-case objective PAR with respect to the uniform distribution.  Appendix~\ref{ap:tpg} discusses average-case objective PAR for a truthful version of the public-good problem.

\section{$2^{nd}$-Price Auctions: Bounds on PARs}\label{section-bounds-II}

In this section, we present upper and lower bounds on the
privacy-approximation ratios for the $2^{nd}$-price Vickrey auction.

\subsection{Problem Specification}

\paragraph*{$2^\mathrm{nd}$-price Vickrey auction.} A single item
is offered to $2$ bidders, each with a private value for the item.
The auctioneer's goal is to allocate the item to the bidder with the
highest value. The fundamental technique in mechanism design for
inducing truthful behavior in single-item auctions is Vickrey's
$2^{nd}$-price auction~\cite{Vic61}: Allocate the item to the
highest bidder, and charge him the second-highest bid.

\begin{definition}
[\twoSPA]
\noindent \\
\noindent \underline{Input:} $x_1,x_2\in \{0,\ldots,2^k-1\}$  (each represented by a $k$-bit string)\\
\noindent \underline{Output:} the identity of the party with the
higher value, \emph{i.e.}, $\arg\max_{i\in \{0,1\}} x_i$ (breaking
ties lexicographically), and the private information of the of the
other party.
\end{definition}

Brandt and Sandholm~\cite{BS} show that a perfectly
privacy-preserving communication protocol exists for \twoSPA.
Specifically, perfect privacy is obtained via the
\emph{ascending-price English auction}: Start with a price of $p=0$
for the item. In each time step, increase $p$ by $1$ until one of
the bidders indicates that his value for the item is less than $p$ (in each step first asking bidder $1$ and then, if necessary, asking bidder $2$).  At that point, allocate the item to the other bidder for a price of
$p-1$. If $p$ reaches a value of $2^k-1$ (that is, the values of
both bidders are $2^k-1$) allocate the item to bidder $1$ for a
price of $2^k-1$.

Moreover, it is shown in~\cite{BS} that the English auction is
essentially \emph{the only} perfectly privacy-preserving protocol
for \twoSPA. Thus, perfect privacy requires, in the worst-case, the
transmission of $\Omega(2^k)$ bits. $2k$ bits suffice, because
bidders can simply reveal their inputs. Can we obtain ``good''
privacy without paying such a high price in communication?

\subsection{Objective Privacy PARs}

We now consider \emph{objective privacy} for \twoSPA (\emph{i.e.},
privacy with respect to the auctioneer). Bisection
auctions~\cite{GHM06,GHMV07} for \twoSPA are defined similarly to
the \bp for \twoMP : Use binary search to find a value $c$ that lies
between the two bidders' values, and let the bidder with the higher
value be bidder $j$. (If the values do not differ, we will also
discover this; in this case, award the item to bidder $1$, who must
pay the common value.) Use binary search on the interval that
contains the value of the lower bidder in order to find the value of
the lower bidder. Bisection auctions are incentive-compatible in ex-post Nash~\cite{GHM06,GHMV07}.

More generally, we refer to an auction protocol as a $c$-bisection
auction, for a constant $c\in (0,1)$, if in each step the interval
$R$ is partitioned into two disjoint subintervals: a lower
subinterval of size $c|R|$ and an upper subinterval of size
$(1-c)|R|$. Hence, the \ba is a $c$-bisection auction with
$c=\frac{1}{2}$. We prove that no $c$-bisection auction for \twoMP
obtains a subexponential objective PAR:

\begin{thm}[A worst-case lower bound for $c$-bisection auctions]\label{thm_worst-case-tradeoff}
\ \ For any constant $c > \frac{1}{2^k}$, the $c$-bisection auction for
\twoSPA has a worst-case PAR of at least
$2^{\frac{k}{2}}$.
\end{thm}

\begin{proof}
Consider the ideal monochromatic partition of \twoSPA depicted for $k=3$ in
Fig.~\ref{fig:illustrate}. Observe that, for perfect privacy to be
preserved, it must be that bidder $2$ transmits the first
(meaningful) bit, and that this bit partitions the space of inputs
into the leftmost shaded rectangle (the set $\{0,\ldots,2^k-1\}\times\{0\}$) and the rest of the value space (ignoring the rectangles depicted that further refine $\{0,\ldots,2^k-1\}\times\{1,\ldots,2^k-1\}$). What if the first bit is
transmitted by player $2$ and does not partition the space into
rectangles in that way? We observe that any other partition of the
space into two rectangles is such that, in the worst case, the
privacy-approximation ratio is at least $2^{\frac{k}{2}}$ (for any
value of $c$): If $c\leq 1-2^{-\frac{k}{2}}$, then the case in which
$x_1=c2^k-1$ gives us the lower bound. If, on the other hand, $c>
1-2^{-\frac{k}{2}}$, then the case that $x_1=0$ gives us the lower
bound. Observe that such a bad PAR is also
the result of bidder $1$'s transmitting the first (meaningful) bit.
\end{proof}

By contrast, as for \twoMP, reasonable privacy
guarantees are achievable in the average case:
\begin{thm}\label{thm:ba-2spa}\rm{\bf{(The average-case objective PAR of the bisection auction)}}\ \
The av\-er\-age-case objective PAR of the \ba is
$\frac{k}{2}+1$ with respect to the uniform distribution.
\end{thm}
\begin{proof}
This follows by taking $g(k)=k$ in Thm.~\ref{thm:2spa-par-bba}.
\end{proof}

We note that the worst-possible approximation of objective privacy
comes when the each value in the space is in a distinct tile; this
is the tiling induced by the sealed-bid auction.  The resulting
average-case privacy-approximation ratio is exponential in $k$.
\begin{prop}[Largest possible objective PAR]
The largest possible (for any protocol) average-case objective
PAR with respect to the uniform distribution
for \twoSPA is
\[
    \frac{1}{2^{2k}}\left[\sum_{j=0}^{2^k-1} j^2  + \sum_{j=1}^{2^k-1} j^2 \right] = \frac{2}{3}2^k + \frac{1}{3}2^{-k}
\]
\end{prop}

\subsection{Subjective Privacy PARs} We now look briefly at subjective
privacy for \twoSPA.  For subjective privacy with respect to $1$, we
start with the $1$-partition for \twoSPA;
Fig.~\ref{fig:2spa:ba1part} shows the refinement of the
$1$-partition induced by the \ba for $k=4$.  Separately considering
the refinement of the $2$-partition for \twoSPA by the \ba, we have
the following results.
\begin{figure}[htp]
\begin{center}
\includegraphics[width=1.5in]{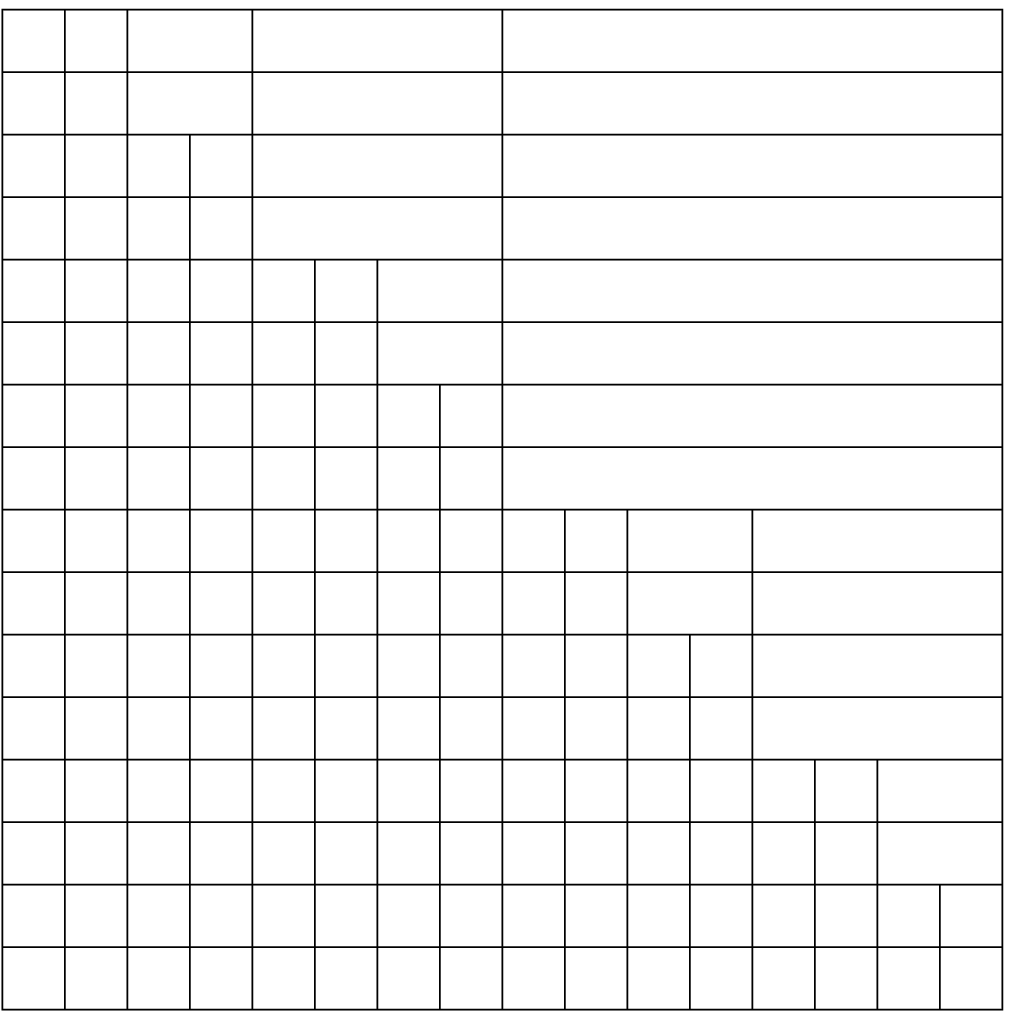}
\caption{The \bha-induced refinement of the $1$-partition for \twoSPA ($k=4$)}\label{fig:2spa:ba1part}
\end{center}
\end{figure}

\begin{thm}\label{thm:ba-2spa-s1}\rm{\bf{(The average-case PAR with respect to $1$ of the bisection auction)}}\ \
The av\-er\-age-case PAR with respect to $1$ of the \ba is \[\frac{k+3}{4}-\frac{k-1}{2^{k+2}}\] with respect to the uniform distribution.
\end{thm}
\begin{proof}
This follows by taking $g(k)=k$ in Thm.~\ref{thm:bba-2spa-s1}.
\end{proof}

\begin{thm}\label{thm:ba-2spa-s2}\rm{\bf{(The average-case PAR with respect to $2$ of the bisection auction)}}\ \
The av\-er\-age-case PAR with respect to $2$ of the \ba is \[\frac{k+5}{4}+\frac{k-1}{2^{k+2}}\] with respect to the uniform distribution.
\end{thm}
\begin{proof}
This follows by taking $g(k)=k$ in Thm.~\ref{thm:bba-2spa-s2}.
\end{proof}
\begin{cor}\label{cor:ba-2spa-s}\rm{\bf{(The average-case subjective PAR of the bisection auction)}}\ \
The av\-er\-age-case subjective PAR of the \ba with respect to the uniform distribution is \[\frac{k+5}{4}+\frac{k-1}{2^{k+2}}.\]
\end{cor}

As for objective privacy, the sealed-bid auction gives the largest possible average-case subjective PAR.
\begin{prop}[Largest possible subjective PAR]
\ The largest possible (for any protocol) average-case subjective PAR with respect to the uniform distribution for \twoSPA is
\[
\frac{2^k}{3} + 1 - \frac{1}{3\cdot 2^k}.
\]
\end{prop}
\begin{proof}
For the sealed-bid auction, the average-case PAR with respect to $1$ is
\[
\frac{1}{2^{2k}}\left[\sum_{j=1}^{2^k} j + \sum_{j=1}^{2^k-1} j^2\right]=\frac{2^k}{3}+\frac{1}{3\cdot 2^{k-1}}.
\]
For the sealed-bid auction, the average-case PAR with respect to $2$ is
\[
\frac{1}{2^{2k}}\left[\sum_{j=1}^{2^k} j^2 + \sum_{j=1}^{2^k-1} j\right] = \frac{2^k}{3} + 1 - \frac{1}{3\cdot 2^k}.
\]
\end{proof}

\subsection{Bounded-Bisection Auctions}

We now present a middle ground between the perfectly-private yet
highly inefficient (in terms of communication) ascending English
auction and the com\-mu\-ni\-ca\-tion-efficient \ba whose average-case
objective PAR is linear in $k$ (and is thus unbounded as $k$ goes to
infinity): We bound the number of bisections, using an ascending
English auction to determine the outcome if it is not resolved by
the limited number of bisections.

We define the \bba as follows: Given an instance of \twoSPA, and a
integer-valued function $g(k)$ such that $0\leq g(k)\leq k$, run the \ba as above but
do at most $g(k)$ bisection operations.  (Note that we will never do more than $k$ bisections.)\ \ If the outcome is
undetermined after $g(k)$ bisection operations, so that both
players' values lie in an interval $I$ of size $2^{k-g(k)}$, apply
the ascending-price English auction to this interval to determine
the identity of the winning bidder and the value of the losing
bidder.

As $g(k)$ ranges from $0$ to $k$, the \bba ranges from the ascending-price
English auction to the \ba.  If we allow a fixed, positive number of
bisections ($g(k)=c>0$), computations show that for $c=1,2,3$ we obtain examples of protocols that do not
provide perfect privacy but that do have bounded average-case
objective PARs with respect to the uniform distribution.  We wish to see if this holds for all positive $c$, determine the average-case objective PAR for general $g(k)$, and connect the amount of
communication needed with the approximation of privacy in this family of protocols.  The following theorem allows us to do these things.

\begin{thm}\label{thm:2spa-par-bba}
For the \bba, the average-case objective PAR with respect to the uniform distribution equals
\[
\frac{g(k)+3}{2}-\frac{2^{g(k)}}{2^{k+1}}+\frac{1}{2^{k+1}}-\frac{1}{2^{g(k)+1}}.
\]
\end{thm}
\begin{proof}
Fix $k$, the number of bits used for bidding, and let $c=g(k)$ be the number of bisections; we have $0\leq c\leq k$, and we let $i=k-c$.  Figure~\ref{fig:2spa-ci} illustrates this tiling for $k=4$, $c=2$, and $i=2$; note that the upper-left and lower-right quadrants have identical structure and that the lower-left and upper-right quadrants have no structure other than that of the ideal partition and the quadrant boundaries (which are induced by the first bisection operation performed).
\begin{figure}[htp]
\begin{center}
\includegraphics[height=1.5in]{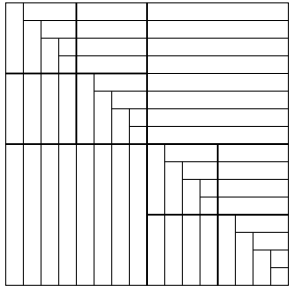}
\end{center}
\caption{Illustration for the proof of Thm.~\ref{thm:2spa-par-bba}}\label{fig:2spa-ci}
\end{figure}

Our general approach is the following.  The average-case objective PAR with respect to the uniform distribution equals
\[
\mathsf{PAR} = \frac{1}{2^{2k}} \sum_{(x_1,x_2)} \frac{|R^I(x_1,x_2)|}{|R^{P}(x_1,x_2)|},
\]
where the sum is over all pairs $(x_1,x_2)$ in the value space; recall that $R^I(x_1,x_2)$ is the region in the ideal partition that contains $(x_1,x_2)$, and $R^P(x_1,x_2)$ is the rectangle in the tiling induced by the protocol that contains $(x_1,x_2)$.  We may combine all of the terms corresponding to points in the same protocol-induced rectangle to obtain
\begin{equation}
\mathsf{PAR} = \frac{1}{2^{2k}} \sum_{R} |R|\frac{|R^I(R)|}{|R|} = \frac{1}{2^{2k}} \sum_{R} |R^I(R)|,\label{eq:par2}
\end{equation}
where the sums are now over protocol-induced rectangles $R$ (we will simplify notation and write $R$ instead of $R^P$), and $R^I(R)$ denotes the ideal region that contains the protocol-induced rectangle $R$.  Each ideal region in which bidder $1$ wins is a rectangle of width $1$ and height at most $2^k$; each ideal region in which bidder $2$ wins is a rectangle of height $1$ and width strictly less than $2^k$.  For a protocol-induced rectangle $R$, let $j_R = 2^k - |R^I(R)|$.  Let $a_{c,i}$ be the total number of tiles that appear in the tiling of the $k$-bit value space induced by the \bba with $g(k)=c$, and let $b_{c,i} = \sum_{R} j_R$ (with this sum being over the protocol-induced tiles in this same partition).  Then we may rewrite~(\ref{eq:par2}) as
\begin{equation}
\mathsf{PAR}_{c,i} = \frac{1}{2^{2k}} \sum_{R} (2^k-j_R) = \frac{a_{c,i}2^k - b_{c,i}}{2^{2k}}.\label{eq:par3}
\end{equation}
(Note that~(\ref{eq:par2}) holds for general protocols; we now add the subscripts ``$c,i$'' to indicate the particular PAR we are computing.)\ \ We now determine $a_{c,i}$ and $b_{c,i}$.

Considering the tiling induced by $c+1$ bisections of a $c+i+1$-bit space (which has $a_{c+1,i}$ total tiles), the upper-left and lower-right quadrants each contain $a_{c,i}$ tiles while the lower-left and upper-right quadrants (as depicted in Fig.~\ref{fig:2spa-ci}) each contribute $2^{c+i}$ tiles, so $a_{c+1,i} = 2a_{c,i}+2^{c+i+1}$.  When there are no bisections, the $i$-bit value space has $a_{0,i}=2^{i+1}-1$ tiles, from which we obtain $a_{c,i} = 2^c \left(2^i (c+2)-1\right)$.  The sum of $j_R$ over protocol-induced rectangles $R$ in the upper-left quadrant is $b_{c,i}$.  For a rectangle $R$ in the lower-right quadrant, $j_R$ equals $2^{c+i}$ plus $j_{R'}$, where $R'$ is the corresponding rectangle in the upper-left quadrant; there are $a_{c,i}$ such $R$, so the sum of $j_R$ over protocol-induced rectangles $R$ in the upper-left quadrant is $b_{c,i}+a_{c,i}2^{c+i}$.  Finally, the sum of $j_R$ over $R$ in the lower-left quadrant equals $\sum_{h=0}^{2^{c+i}-1} h$ and the sum over $R$ in the top-right quadrant equals $\sum_{h=1}^{2^{c+i}} h$.  Thus, $b_{c+1,i} = 2b_{c,i}+a_{c,i}2^{c+i} + 2^{2(c+i)}$; with $b_{0,i} = \sum_{h=0}^{2^i-1} h + \sum_{h=1}^{2^i-1} h$, we obtain $b_{c,i} = 2^{c+i-1} \left(\left(1+2^c\right) \left(-1+2^i\right)+2^{c+i} c\right)$.  Rewriting~(\ref{eq:par3}), we obtain
\[
\mathsf{PAR}_{c,i} = \frac{c+3}{2}-\frac{2^{c}}{2^{c+i+1}}+\frac{1}{2^{c+i+1}}-\frac{1}{2^{c+1}}.
\]
Recalling that $k=c+i$, the proof is complete.
\end{proof}

For the protocols corresponding to values of $g(k)$
ranging from $0$ to $k$ (ranging from the ascending-price
English auction to the \ba), we may thus relate the amount of
communication saved (relative to the English auction) to the effect of this
on the PAR.
\begin{cor}\label{cor:bba}
Let $g$ be a function that maps non-negative integers to non-negative
integers. Then the average-case objective PAR with respect to the
uniform distribution for the \bba is bounded
if $g$ is bounded and is unbounded if $g$ is unbounded.
We then have that the \bba may require the exchange of $\Theta(k+2^{k-g(k)})$
bits, and it has an average-case objective PAR of $\Theta(1+g(k))$.
\end{cor}

\begin{rem}
Some of the sequences that appear in the proof above also appear in other settings.  For example, the sequences $\{a_{0,i}\}_i$, $\{a_{1,i}\}_i$, and $\{a_{2,i}\}_i$ are slightly shifted versions of sequences A000225, A033484, and A028399, respectively, in the OEIS~\cite{oeis}, which notes other combinatorial interpretations of them.
\end{rem}

\subsubsection{Subjective privacy for bounded-bisection auctions}

\begin{thm}\label{thm:bba-2spa-s1}\rm{\bf{(The average-case PAR w.r.t.\ $1$ of the bounded-bi\-section auction)}}\ \
The average-case PAR with respect to $1$ of the \bba is \[\frac{g(k)+5}{4}-\frac{1}{2^{g(k)+2}} - \frac{1}{2^{k-g(k)+1}} - \frac{g(k)-2}{2^{k+2}}\] with respect to the uniform distribution.
\end{thm}
\begin{proof}
The approach is similar to that in the proof of Thm.~\ref{thm:2spa-par-bba}.  We start by specializing~(\ref{eq:par2}) to the present case.

Each ideal region in which bidder $1$ wins is a rectangle of size $1$; each ideal region in which bidder $2$ wins is a rectangle of height $1$ and width strictly less than $2^k$.  For a protocol-induced rectangle $R$, let $j_R = 2^k - |R^I(R)|$.  Let $c=g(k)$ and let $i=k-c\geq 0$.  Let $T_{c,i}^1$ be the refinement of the \twoSPhA $1$-partition of the $k$-bit value space induced by the \bhba.  Let $x_{c,i}$ be the number of rectangles in $T_{c,i}^1$ in which bidder $2$ (the column player) wins, and let $y_{c,i}$ be the sum, over all rectangles $R$ in which bidder $2$ wins, of the quantity $2^{c+i}-|R^I(R)|$.  Let $z_{c,i}$ be the number of rectangles $R$ in which bidder $1$ (the row player) wins.

Using $\mathsf{PAR}^1_{c,i}$ to denote the PAR w.r.t.\ bidder $1$ in this case ($c$ bisections and $i=k-c$), we may rewrite~(\ref{eq:par2}) as
\begin{eqnarray*}
\mathsf{PAR}^1_{c,i} &=& \frac{1}{2^{2(c+i)}} \left[\left(\sum_{\substack{R^P\mathrm{\ in\ which}\\ 1\mathrm{\ wins}}}
|R^I(R^P)|\right) +  \left(\sum_{\substack{R^P\mathrm{\ in\ which}\\ 2\mathrm{\ wins}}}
|R^I(R^P)|\right)\right]\\
 &=&\frac{1}{2^{2(c+i)}}\left[\left(z_{c,i}\right)+\left(2^{c+i}x_{c,i} - y_{c,i}\right)\right].
\end{eqnarray*}
We now turn to the computation of $x_{c,i}$, $y_{c,i}$, and $z_{c,i}$.

Following the same approach as in the proof of Thm.~\ref{thm:2spa-par-bba}, we have $x_{c+1,i} = 2x_{c,i} + 2^{c+i}$, $y_{c+1,i}=2 y_{c,i} + \sum_{j=1}^{2^{c+i}} j + 2^{c+i}x_{c,i}$, and $z_{c+1,i} = 2z_{c,i} + 2^{2(c+i)}$.  With $x_{0,i}=2^i-1$, $y_{0,i}=\sum_{j=1}^{2^i-1} j$, and $z_{0,i} = \sum_{j=1}^{2^{c+1}} j$, we obtain
\begin{eqnarray*}
x_{c,i} &=& 2^{c-1}\left(2^i c + 2^{i+1} - 2\right),\\
y_{c,i} &=& 2^{c+i-2}\left(2^{c+i}c +2^{c+i} +2^i - 2^{c+1} + c\right),\mathrm{\ and}\\
z_{c,i} &=& 2^{c+i-1}(2^{c+i}+1).
\end{eqnarray*}
Using these in our expression for $\mathsf{PAR}^1_{c,i}$, we obtain
\[
\mathsf{PAR}^1_{c,i} = \frac{c+5}{4} + \frac{2-c}{2^{c+i+2}} - \frac{1}{2^{i+1}} - \frac{1}{2^{c+2}}.
\]
Recalling that $k=c+i$ and $g(k)=c$ completes the proof.
\end{proof}

\begin{thm}\label{thm:bba-2spa-s2}\rm{\bf{(The average-case PAR w.r.t.\ $2$ of the bounded-bi\-section auction)}}\ \
The average-case PAR with respect to $1$ of the \bba is \[\frac{g(k)+5}{4}-\frac{1}{2^{g(k)+2}} + \frac{g(k)}{2^{k+2}}\] with respect to the uniform distribution.
\end{thm}
\begin{proof}
The approach is essentially the same as in the proof of Thm.~\ref{thm:bba-2spa-s1}, although the induced partition differs slightly.

Let $c=g(k)$ and let $i=k-c\geq 0$.  Let $T_{c,i}^2$ be the refinement of the \twoSPhA $2$-partition of the $k$-bit value space induced by the \bhba.  Let $u_{c,i}$ be the number of rectangles in $T_{c,i}^2$ in which bidder $1$ (the row player) wins, and let $v_{c,i}$ be the sum, over all rectangles $R$ in which bidder $1$ wins, of the quantity $2^{c+i}-|R^I(R)|$.  Let $w_{c,i}$ be the number of rectangles $R$ in which bidder $2$ (the column player) wins.  Using $\mathsf{PAR}^2_{c,i}$ to denote the PAR w.r.t.\ bidder $2$ in this case ($c$ bisections and $i=k-c$), we may rewrite~(\ref{eq:par2}) as
\[
\mathsf{PAR}^2_{c,i} = \frac{1}{2^{2(c+i)}}\left[\left(2^{c+i} u_{c,i} - v_{c,i}\right)+\left(w_{c,i}\right)\right].
\]

Mirroring the approach of the proof of Thm.~\ref{thm:bba-2spa-s1}, we have $u_{c+1,i} = 2u_{c,i} + 2^{c+i}$, $v_{c+1,i} = 2v_{c,i} + 2^{c+i-1}(2^{c+i}-1) + 2^{c+i}u_{c,i}$, and $w_{c,i} = 2^{c+i-1}(2^{c+i}-1)$.  With $u_{0,i} = 2^i$ and $v_{0,i} = 2^{i-1}(2^i-1)$, we obtain
\begin{eqnarray*}
u_{c,i} &=& 2^{c+i-1}(c+2),\\
v_{c,i} &=& 2^{c+i-2}\left(2^{c+i}(c+1)+2^i-c-2\right),\mathrm{\ and}\\
w_{c,i} &=& 2^{c+i-1}(2^{c+i}-1).
\end{eqnarray*}
Using these in our expression for $\mathsf{PAR}^2_{c,i}$, we obtain
\[
\mathsf{PAR}^2_{c,i} = \frac{c+5}{4} - \frac{1}{2^{c+2}} + \frac{c}{2^{c+i+2}}.
\]
Recalling that $k=c+i$ and $g(k)=c$ completes the proof.
\end{proof}

Because $g(k)\geq 0$, the average-case PAR with respect to $2$ is at least as large as the average-case PAR with respect to $1$; this gives the average-case subjective PAR of the \bba as follows.
\begin{cor}\label{cor:bba-2spa-sub}\rm{\bf{(Av\-er\-age-case subjective PAR of the bounded-bi\-section auction)}}\ \
The av\-er\-age-case subjective PAR of the \bba is \[\frac{g(k)+5}{4}-\frac{1}{2^{g(k)+2}} + \frac{g(k)}{2^{k+2}}\] with respect to the uniform distribution.
\end{cor}

Finally, Table~\ref{tab:2spa} summarizes the average-case PAR results (with respect to the uniform distribution) for \twoSPA.
\begin{table}[htp]
\begin{center}
\begin{tabular}{|c|c|c|}
  \hline
   & Avg.-Case Obj.\ PAR & Avg.-Case Subj.\ PAR \\ \hline
  English Auction & $1$ & $1$ \\ \hline
  \bba & $\frac{g(k)+3}{2}-\frac{2^{g(k)}}{2^{k+1}}+$ & $\frac{g(k)+5}{4}-\frac{1}{2^{g(k)+2}} + $ \\
   & $\qquad\frac{1}{2^{k+1}}-\frac{1}{2^{g(k)+1}}$ & $\qquad\qquad\frac{g(k)}{2^{k+2}}$\\ \hline
  \ba & $\frac{k}{2}+1$ & $\frac{k+5}{4}+\frac{k-1}{2^{k+2}}$\\ \hline
  Sealed-Bid Auction & $\frac{2^{k+1}}{3}+\frac{1}{3\cdot 2^k}$ & $\frac{2^k}{3}+1-\frac{1}{3\cdot 2^k}$\\
  \hline
\end{tabular}
\end{center}
\caption{Average-case PARs (with respect to the uniform distribution) for \twoSPA}\label{tab:2spa}
\end{table}

\section{Discussion and Future Directions}\label{section-discussion}

\subsection{Other Notions of Approximate Privacy}

By our definitions, the worst-case/average-case PARs of a protocol
are determined by the worst-case/expected value of the expression
$\frac{|R^I(\mathbf{x})|}{|R^P(\mathbf{x})|}$, where $R^P(\mathbf{x})$ is the
monochromatic rectangle induced by $P$ for input $\mathbf{x}$, and
$R^I(\mathbf{x})$ is the monochromatic region containing
$A(f)_{\mathbf{x}}$ in the ideal monochromatic partition of $A(f)$.
That is, informally, we are interested in the ratio of the
\emph{size} of the ideal monochromatic region for a specific pair of inputs to the \emph{size} of the monochromatic rectangle induced by the protocol for that pair. More generally, we can define worst-case/average-case PARs with respect to a function $g$ by considering the ratio
$\frac{g(R^I(\mathbf{x}),\mathbf{x})}{g(R^P(\mathbf{x}),\mathbf{x})}$.  Our definitions of PARs set $g(R,\mathbf{x})$ to be the cardinality of $R$.  This captures the intuitive notion of the
indistinguishability of inputs that is natural to consider in the
context of privacy preservation. Other definitions of PARs may be
appropriate in analyzing other notions of privacy.  We suggest a few here; further
investigation of these and other definitions provides many interesting avenues for future work.

\textbf{Probability mass.}\ \ Given a probability distribution $D$ over the parties' inputs, a
seemingly natural choice of $g$ is the probability mass. That is,
for any region $R$, $g(R)=Pr_D(R)$, the probability (according to
$D$) that the input corresponds to an entry in $R$. However, a
simple example illustrates that this intuitive choice of
$g$ is problematic: Consider a problem for which $\{0,\ldots,n\}\times\{i\}$
is a maximal monochromatic region for $0\leq i\leq n-1$ as illustrated in the left part of Fig.~\ref{fig:prob-mass}.  Let $P$ be the communication protocol consisting
of a single round in which party $1$ reveals whether or not his value is $0$;
this induces the monochromatic tiling with tiles $\{(0,i)\}$ and $\{(1,i),\ldots,(n,i)\}$ for each $i$ as illustrated in the right part of Fig.~\ref{fig:prob-mass}.  Now, let $D_1$ and $D_2$ be the probability distributions over the inputs $\mathbf{x}=(x_1,x_2)$ such that, for $0\leq i\leq n-1$ and $1\leq j\leq n$, $Pr_{D_1}[(x_1,x_2)=(0,i)]=\frac{\epsilon}{n}$, $Pr_{D_1}[(x_1,x_2)=(j,i)]=\frac{1-\epsilon}{n^2}$, $Pr_{D_2}[(x_1,x_2)=(0,i)]=\frac{1-\epsilon}{n}$, and $Pr_{D_2}[(x_1,x_2)=(j,i)]=\frac{\epsilon}{n^2}$ for some small $\epsilon>0$. Intuitively, any
reasonable definition of PAR should imply that, for $D_1$, $P$
provides ``bad'' privacy guarantees (because w.h.p.~it reveals the
value of $x_1$),
and, for $D_2$, $P$ provides ``good''
privacy (because w.h.p.~it reveals little
about $x_1$).
In sharp contrast, choosing $g$ to be the probability mass results
in the same average-case PAR in both cases.

\begin{figure}[htp]
\begin{center}
\includegraphics[width=4in]{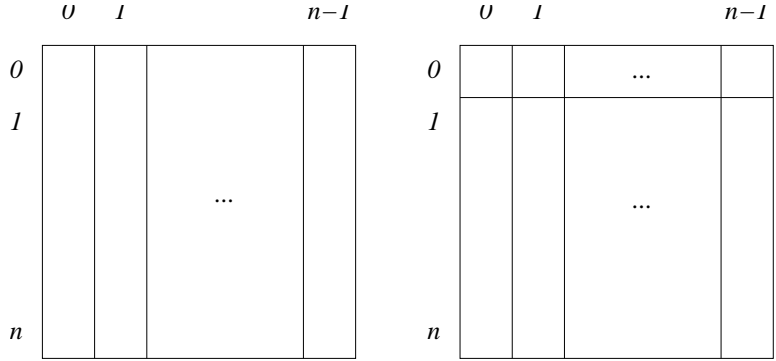}
\caption{\small Maximal monochromatic regions (left) and protocol-induced rectangles (right) for an example showing the deficiencies of PAR definitions based on probability mass. }\label{fig:prob-mass}
\end{center}
\end{figure}

\textbf{Other additive functions.}\ \ In our definition of PAR and in the probability-mass approach, each input $\mathbf{x}$ in a rectangle contributes to $g(R,\mathbf{x})$ in a way that is independent of the other inputs in $R$.  Below, we discuss some natural approaches that violate this condition, but we start by noting that other functions that satisfy this condition may be of interest.  For example, taking $g(R,\mathbf{x}) = 1+\sum_{\mathbf{y}\in R\setminus \mathbf{x}} d(\mathbf{x},\mathbf{y})$, where $d$ is some distance defined on the input space, gives our original definition of PAR when $d(x,y) = 1-\delta_{\mathbf{x},\mathbf{y}}$ and might capture other interesting definitions (in which indistinguishable inputs that are farther away from $\mathbf{x}$ contribute more to the privacy for $\mathbf{x}$).  (The addition of $1$ ensures that the ratio $g(R^I,\mathbf{x})/g(R^P,\mathbf{x})$ is defined, but that can be accomplished in other ways if needed.)\ \ Importantly, here and below, the notion of distance that is used might not be a Euclidean metric on the $n$-player input space $[0,2^k-1]^n$.  It could instead (and likely would) focus on the problem-specific interpretation of the input space.  Of course, there are may possible variations on this (\eg, also accounting for the probability mass).

\textbf{Maximum distance.}\ \ We might take the view that a protocol does not reveal much about an input $\mathbf{x}$ if there is another input that is ``very different'' from $\mathbf{x}$ that the protocol cannot distinguish from $\mathbf{x}$ (even if the total number of things that are indistinguishable from $\mathbf{x}$ under the protocol is relatively small).  For some distance $d$ on the input space, we might than take $g$ to be something like $1+\max_{\mathbf{y}\in R\setminus\{\mathbf{x}\}} d(\mathbf{y},\mathbf{x})$.

\textbf{Plausible deniability.}\ \ One drawback to the maximum-distance approach is that it does not account for the probability associated with inputs that are far from $\mathbf{x}$ (according to a distance $d$) and that are indistinguishable from $\mathbf{x}$ under the protocol.  While there might be an input $\mathbf{y}$ that is far away from $\mathbf{x}$ and indistinguishable from $\mathbf{x}$, the probability of $\mathbf{y}$ might be so small that the observer feels comfortable assuming that $\mathbf{y}$ does not occur.  A more realistic approach might be one of ``plausible deniability.''\ \ This makes use of a plausibility threshold---intuitively, the minimum probability that the ``far away'' inputs(s) (which is/are indistinguishable from $\mathbf{x}$) must be assigned in order to ``distract'' the observer from the true input $\mathbf{x}$.  This threshold might correspond to, \eg, ``reasonable doubt'' or other levels of certainty.  We then consider how far we can move away from $\mathbf{x}$ while still having ``enough'' mass (\ie, more than the plausibility threshold) associated with the elements indistinguishable from $\mathbf{x}$ that are still farther away.  We could then take $g$ to be something like $1+\max \{d_0 | Pr_D(\{\mathbf{y}\in R|d(\mathbf{y},\mathbf{x})\geq d_0\})/Pr_D(R) \geq t\}$; other variations might focus on mass that is concentrated in a particular direction from $\mathbf{x}$.  (In quantifying privacy, we would expect to only consider those $R$ with positive probability, in which case dividing by $Pr_D(R)$ would not be problematic.)\ \ Here we use $Pr_D(R)$ to normalize the weight that is far away from $x$ before comparing it to the threshold $t$; intuitively, an observer would know that the value is in the same region as $x$, and so this seems to make the most sense.

\textbf{Relative rectangle size.}\ \ One observation is that a bidder likely has a very different view of an auctioneer's being able to tell (when some particular protocol is used) whether his bid lies between $995$ and $1005$ than he does of the auctioneer's being able to tell whether his bid lies between $5$ and $15$.  In each case, however, the bids in the relevant range are indistinguishable under the protocol from $11$ possible bids.  In particular, the privacy gained from an input's being distinguishable from a fixed number of other inputs may (or may not) depend on the context of the problem and the intended interpretation of the values in the input space.  This might lead to a choice of $g$ such as $diam_d(R)/|\mathbf{x}|$, where $diam_d$ is the diameter of $R$ with respect to some distance $d$ and $|\mathbf{x}|$ is some (problem-specific) measure of the size of $\mathbf{x}$ (\eg, bid value in an auction).  Numerous variations on this are natural and may be worth investigating.

\textbf{Information-theoretic approaches.}\ \ Information-the\-oretic approaches using conditional entropy are also natural to consider when studying privacy, and these have been used in various settings.  Most relevantly, Bar-Yehuda \etal.~\cite{BCKO} defined multiple measures based on the conditional mutual information about one player's value (viewed as a random variable) revealed by the protocol trace and knowledge of the other player's value.  It would also be natural to study objective-PAR versions using the entropy of the random variable corresponding to the (multi-player) input conditioned only on the protocol output (and not the input of any player).  Such approaches might facilitate the comparison of privacy between different problems.

\subsection{Open Questions}

There are many interesting directions for future research:

\begin{itemize}

\item  As discussed in the previous subsection, the definition and
exploration of other notions of PARs is a challenging and intriguing
direction for future work.

\item We have shown that, for both \twoMP and \twoSPA, reasonable average-case PARs
with respect to the uniform distribution are achievable. We
conjecture that our upper bounds for these problems extend to
\emph{all} possible distributions over inputs.

\item An interesting open question is proving lower bounds on the
average-case PARs for \twoMP and \twoSPA.

\item It would be interesting to apply the PAR framework presented in this paper
to other functions.

\item The extension of our PAR framework to the $n$-party communication
model is a challenging direction for future research.

\item Starting from the same place that we did,
namely~\cite{CK91,K92}, Bar-Yehuda {\it et al.}~\cite{BCKO} provided
three definitions of approximate privacy. The one that seems most
relevant to the study of privacy-approximation ratios is their
notion of {\it h-privacy}. It would be interesting to know exactly
when and how it is possible to express PARs in terms
of $h$-privacy and {\it vice versa}.

\end{itemize}

\section*{Acknowledgements}

We are grateful to audiences at University Residential Centre of Bertinoro, Boston University, DIMACS, the University of Massach\-usetts, Northwestern, Princeton, and Rutgers for helpful questions and feedback.

\appendix

\section{Relation to the Work of Bar-Yehuda \emph{et al.}~\cite{BCKO}}\label{apx:bar-yehuda}

While there are certainly some parallels between the work here and the
Bar-Yehuda \etal.\ work~\cite{BCKO}, there are significant differences in
what the two frameworks capture. Specifically:

\begin{enumerate}

\item The results in~\cite{BCKO} deal with what can be learned by a party who
knows one of the inputs. By contrast, our notion of objective
PAR captures the effect of a protocol on privacy with respect
to an external observer who does not know any of the players
private values.

\item More importantly, the framework of~\cite{BCKO} does not address the size of
monochromatic regions. As illustrated by the following example,
the ability to do so is necessary to capture the effects of
protocols on interesting aspects of privacy that are captured by
our definitions of PAR.

Consider the function $f:\{0,\ldots,2^n - 1\}\times\{0,\ldots,2^n - 1\}\rightarrow\{0,\ldots,2^{n-2}\}$
defined by $f(x,y) = floor(\frac{x}{2})$ if $x < 2^{n-1}$ and $f(x,y) =
2^{n-2}$ otherwise. Consider the following two protocols for $f$:
in $P$, player $1$ announces his value $x$ if $x < 2^{n-1}$ and
otherwise sends $2^{n-1}$ (which indicates that $f(x,y) =
2^{n-2}$); in $Q$, player $1$ announces $floor(\frac{x}{2})$ if $x < 2^n - 1$
and $x$ if $x = 2^n - 1$. Observe that each protocol induces
$2^{n-1} + 1$ rectangles.

Intuitively, the effect on privacy of these two protocols is
different. For half of the inputs, $P$ reduces by a factor of $2$
the number of inputs from which they are indistinguishable
while not affecting the indistinguishability of the other
inputs. $Q$ does not affect the indistinguishability of the
inputs affected by $P$, but it does reduce the number of inputs
indistinguishable from a given input with $x\geq 2^{n-1}$ by at
least a factor of $2^{n-2}$.

Our notion of PAR is able to capture the different effects on
privacy of the protocols $P$ and $Q$. (The average-case objective
PARs are constant and exponential in $n$, respectively.) By
contrast, the three quantifications of privacy from~\cite{BCKO}---$I_c$,
$I_i$, and $I_{c-i}$---do not distinguish between these two
protocols; we now sketch the arguments for this claim.

For each protocol, any function $h$ for which the protocol is
weakly $h$-private must take at least $2^{n-1} + 1$ different
values. This bound is tight for both $P$ and $Q$. Thus, $I_c$ cannot
distinguish between the effects of $P$ and $Q$ on $f$.

The number of rectangles induced by $P$ that intersect each row
and column equals the number induced by $Q$. Considering the
geometric interpretation of $I_P$ and $I_Q$, as well as the
discussion in Sec. VII.A of~\cite{BCKO}, we see that $I_i$ and $I_{c-i}$
(applied to protocols) cannot distinguish between the effects
of $P$ and $Q$ on $f$.
\end{enumerate}

\section{Truthful Public-Good Problem}\label{ap:tpg}

\subsection{Problem}

As in Sec.~\ref{ssec:pg}, the government is considering the construction of a bridge at cost $c$.  Each taxpayer has a private value that is the utility he would gain from the bridge if it were built, and the government wants to build the bridge if and only if the sum of the taxpayers' private values is at least $c$.  Now, in addition to determining whether to build the bridge, the government incentivizes truthful disclosure of the private values by requiring taxpayer $i$ to pay $c-\sum_{j\neq i} x_j$ if $\sum_{j\neq i} x_j < c$ but $\sum_i x_i \geq c$ (see, \eg,~\cite{nisan07agt} for a discussion of this type of approach).  The government should thus learn whether or not to build the bridge and how much, if anything, each taxpayer should pay.  The formal description of the function is as follows; the corresponding ideal partition of the value space is shown in Fig.~\ref{fig:tpgkc-k3c4-ideal}, in which regions for which the output is ``Build'' are just labeled with the appropriate value of $(t_1,t_2)$.

\begin{definition}
[\tpgkc]
\noindent \\
\noindent \underline{Input:} $c,x_1,x_2\in \{0,\ldots,2^k-1\}$  (each represented by a $k$-bit string)\\
\noindent \underline{Output:} ``Do Not Build'' if $x_1 + x_2 < c$; ``Build'' and $(t_1,t_2)$ if $x_1 + x_2 \geq c$, where $t_i = c-x_{3-i}$ if $x_{3-i}<c$ and $x_1+x_2 \geq c$, and $t_i=0$ otherwise.
\end{definition}

\begin{figure}[htp]
\begin{center}
\includegraphics[height=2in]{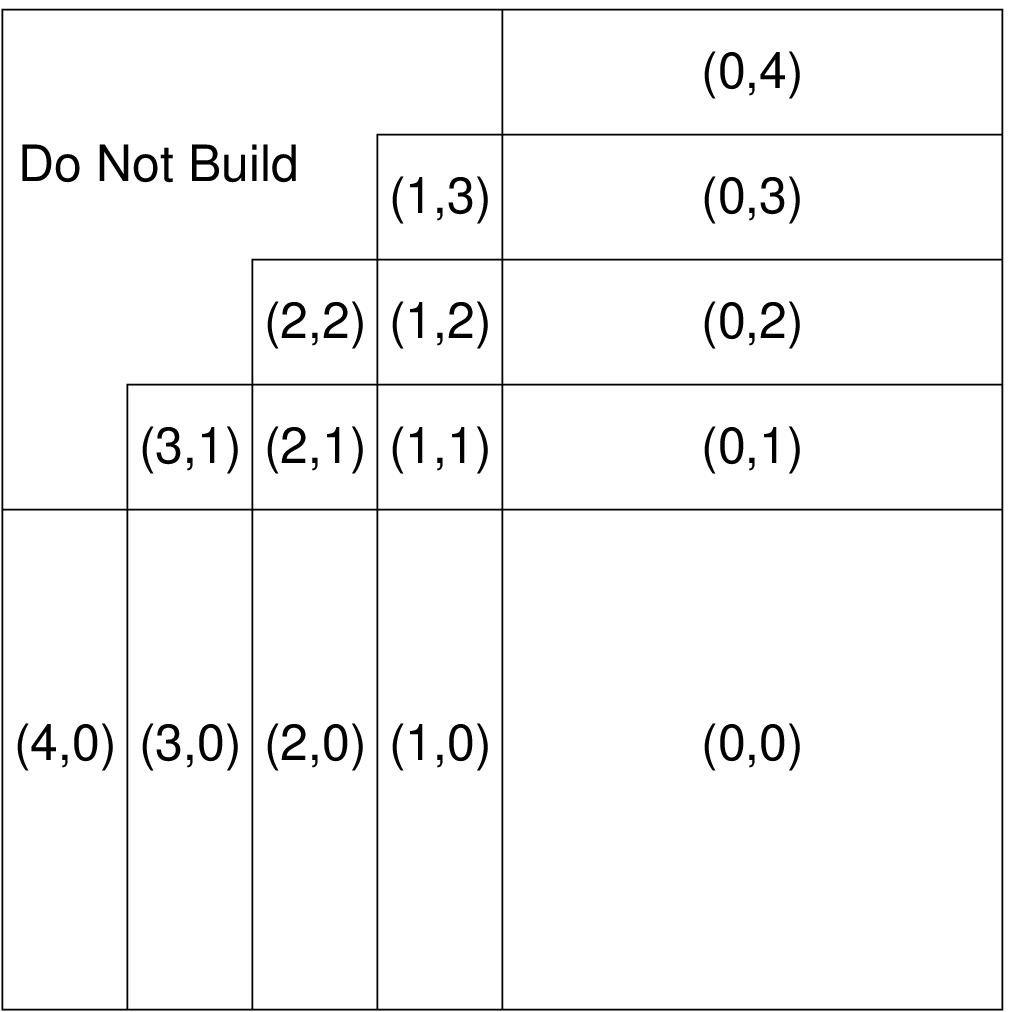}
\end{center}
\caption{Ideal partition of the value space for \tpgkc with $k=3$ and $c=4$.}\label{fig:tpgkc-k3c4-ideal}
\end{figure}

\subsection{Results}

\begin{prop}\rm{\bf{(Average-case objective PAR of \tpgkc)}}\ \
The average-case objective PAR of \tpgkc with respect to the uniform distribution is
\[
1 + \frac{c^3}{2^{2k+1}}(1-\frac{1}{c^2}).
\]
\end{prop}
\begin{proof}
We may rewrite Eq.~\ref{eq:par2} as (adding subscripts for the values of $k$ and $c$ in this problem):
\[
\mathsf{PAR}_{k,c} = \frac{1}{2^{2k}} \left[\sum_{R_\mathsf{DNB}} |R^I(R_\mathsf{DNB})| + \sum_{R_\mathsf{B}} |R^I(R_\mathsf{B})|\right],
\]
where the first sum is taken over rectangles $R_\mathsf{DNB}$ for which the output is ``Do Not Build'' and the second sum is taken over rectangles $R_\mathsf{B}$ for which the output is ``Build'' together with some $(t_1,t_2)$.  Using the same argument as for \twoMP, the first sum must be taken over at least $c$ rectangles; the ideal region containing these rectangles has size $\sum_{i=1}^{c} i = c(c+1)/2$.  Considering the second sum, each of the ideal regions containing a protocol-induced rectangle is in fact a rectangle.  If the protocol did not further partition these rectangles (and it is easy to see that such protocols exist) then the total contribution of the second sum is just the total number of inputs for which the output is ``Do Not Build'' together with some pair $(t_1,t_2)$, \ie, this contribution is $4^k - c(c+1)/2$.  We may thus rewrite $\mathsf{PAR}_{k,c}$ as
\[
\mathsf{PAR}_{k,c} = \frac{1}{2^{2k}} \left[ c\frac{c(c+1)}{2} + 4^k - \frac{c(c+1)}{2}\right] = 1 + \frac{c^3}{2^{2k+1}}(1-\frac{1}{c^2})
\]
\end{proof}

Unsurprisingly, if we take $c=2^k-1$ (as in \pgk in Sec.~\ref{ssec:pg}), we obtain $\mathsf{PAR}_{k,2^k-1} = 2^{k-1}-\frac{1}{2}+\frac{1}{2^k}$, which is essentially half of the average-case PAR for \twoMP.


\begin{thebibliography}{10}

\bibitem{BBNS08}
Moshe Babaioff, Liad Blumrosen, Moni Naor, and Michael Schapira.
\newblock Informational overhead of incentive compatibility.
\newblock In {\em Proceedings of the ACM conference on Electronic commerce},
  pages 88--97, 2008.

\bibitem{BCKO}
Reuven Bar-Yehuda, Benny Chor, Eyal Kushilevitz, and Alon Orlitsky.
\newblock Privacy, additional information, and communication.
\newblock {\em IEEE Transactions on Information Theory}, 39:55--65, 1993.

\bibitem{BCNW}
Amos Beimel, Paz Carmi, Kobbi Nissim, and Enav Weinreb.
\newblock Private approximation of search problems.
\newblock In {\em Proceedings of the ACM Symposium on Theory of Computing},
  pages 119--128, 2006.

\bibitem{BGW88}
Michael Ben-Or, Shafi Goldwasser, and Avi Wigderson.
\newblock Completeness theorems for non-cryptographic, fault-tolerant
  computation.
\newblock In {\em Proceedings of the ACM Symposium on Theory of Computing},
  pages 1--10, 1988.

\bibitem{BN02}
Liad Blumrosen and Noam Nisan.
\newblock Auctions with severely bounded communications.
\newblock In {\em Proceedings of the IEEE Symposium on Foundations of Computer
  Science}, pages 406--415, 2002.

\bibitem{BS}
Felix Brandt and Tuomas Sandholm.
\newblock On the existence of unconditionally privacy-preserving auction
  protocols.
\newblock {\em ACM Trans. Inf. Syst. Secur.}, 11(2):1--21, 2008.

\bibitem{CCD88}
David Chaum, Claude Cr\'epeau, and Ivan Damgaard.
\newblock Multiparty, unconditionally secure protocols.
\newblock In {\em Proceedings of the ACM Symposium on Theory of Computing},
  pages 11--19, 1988.

\bibitem{CK91}
Benny Chor and Eyal Kushilevitz.
\newblock A zero-one law for boolean privacy.
\newblock {\em SIAM J. Discrete Math}, 4:36--47, 1991.

\bibitem{DHR00}
Yevgeniy Dodis, Shai Halevi, and Tal Rabin.
\newblock A cryptographic solution to a game-theoretic problem.
\newblock In {\em Advances in Cryptology---Crypto 2000}, Springer Verlag LNCS.

\bibitem{DiffPrivSurvey}
Cynthia Dwork.
\newblock Differential privacy.
\newblock In {\em Proceedings of the International Colloquium on Automata,
  Languages and Programming}, pages 1--12, 2006.

\bibitem{FIMNSW}
Joan Feigenbaum, Yuval Ishai, Tal Malkin, Kobbi Nissim, Martin~J. Strauss, and
  Rebecca~N. Wright.
\newblock Secure multiparty computation of approximations.
\newblock {\em ACM Transactions on Algorithms}, 2(3):435--472, 2006.

\bibitem{FS02}
Joan Feigenbaum and Scott Shenker.
\newblock Distributed algorithmic mechanism design: Recent results and future
  directions.
\newblock In {\em Proceedings of the ACM Mobicom Workshop on Discrete
  Algorithms and Methods for Mobile Computing and Communications}, pages 1--13,
  2002.

\bibitem{GRS09}
Arpita Ghosh, Tim Roughgarden, and Mukund Sundararajan.
\newblock Universally utility-maximizing privacy mechanisms.
\newblock In {\em Proceedings of the ACM Symposium on Theory of Computing},
  pages 351--360, 2009.

\bibitem{GHM06}
Elena Grigorieva, P.~Jean-Jacques Herings, and Rudolf M{\"u}ller.
\newblock The communication complexity of private value single-item auctions.
\newblock {\em Oper. Res. Lett.}, 34(5):491--498, 2006.

\bibitem{GHMV07}
Elena Grigorieva, P.~Jean-Jacques Herings, Rudolf M{\"u}ller, and Dries
  Vermeulen.
\newblock The private value single item bisection auction.
\newblock {\em Economic Theory}, 30(1):107--118, January 2007.

\bibitem{HKKN}
Shai Halevi, Robert Krauthgamer, Eyal Kushilevitz, and Kobbi Nissim.
\newblock Private approximation of {NP}-hard functions.
\newblock In {\em Proceedings of the ACM Symposium on Theory of Computing},
  pages 550--559, 2001.

\bibitem{K92}
Eyal Kushilevitz.
\newblock Privacy and communication complexity.
\newblock {\em SIAM J. Discrete Math.}, 5(2):273--284, 1992.

\bibitem{KN97}
Eyal Kushilevitz and Noam Nisan.
\newblock {\em Communication Complexity}.
\newblock Cambridge University Press, 1997.

\bibitem{NPS99}
Moni Naor, Benny Pinkas, and Reuben Sumner.
\newblock Privacy preserving auctions and mechanism design.
\newblock In {\em Proceedings of the ACM Conference on Electronic Commerce},
  pages 129--139, 1999.

\bibitem{nisan07agt}
Noam Nisan.
\newblock Introduction to mechanism design (for computer scientists).
\newblock In Noam Nisan, Tim Roughgarden, \'{E}va Tardos, and Vijay~V.
  Vazirani, editors, {\em Algorithmic Game Theory}, chapter~9. Cambridge
  University Press, 2007.

\bibitem{oeis}
The On-Line Encyclopedia of Integer Sequences.
\newblock Published electronically at \url{http://oeis.org}, 2010.

\bibitem{Vic61}
William Vickrey.
\newblock Counterspeculation, auctions and competitive sealed tenders.
\newblock {\em Journal of Finance}, pages 8--37, 1961.

\bibitem{Y79}
Andrew~C. Yao.
\newblock Some complexity questions related to distributive computing
  (preliminary report).
\newblock In {\em Proceedings of the ACM Symposium on Theory of Computing},
  pages 209--213, 1979.

\bibitem{Yao-millionaires}
Andrew~C. Yao.
\newblock Protocols for secure computation.
\newblock In {\em Proceedings of the IEEE Symposium on Foundations of Computer
  Science}, pages 160--164, 1982.

\bibitem{Yao-general}
Andrew~C. Yao.
\newblock How to generate and exchange secrets.
\newblock In {\em Proceedings of the IEEE Symposium on Foundations of Computer
  Science}, pages 162--167, 1986.

\end{thebibliography}
\end{document}